\documentclass[a4paper,onecolumn,11pt,accepted=2025-03-03]{quantumarticle}
\pdfoutput=1
\usepackage[utf8]{inputenc}
\usepackage[english]{babel}
\usepackage[T1]{fontenc}
\usepackage{amsmath}
\usepackage[numbers,sort&compress]{natbib}
\usepackage{adjustbox}

\newcommand{\hi}{\mathcal{H}}

\newcommand{\hio}{\mathcal{H}_{\mathcal{R}_1}}
\newcommand{\hit}{\mathcal{H}_{\mathcal{R}_2}}
\newcommand{\hith}{\mathcal{H}_{\mathcal{R}_3}}
\newcommand{\his}{\mathcal{H}_{\mathcal{S}}}
\newcommand{\hir}
{\mathcal{H}_{\mathcal{R}}}

\newcommand{\Eff}{\mathcal{E}}

\newcommand{\Y}{\yen}

\newcommand{\ip}[2]{\left\langle\,#1\,|\,#2\,\right\rangle} 
\newcommand{\ket}[1]{|#1\rangle} 

\newcommand{\bra}[1]{\langle#1|} 
\newcommand{\no}[1]{\left\|#1\right\|} 
 
\newcommand{\tr}[1]{\textrm{tr}\left[#1\right]} 
\newcommand{\id}{\mathbbm{1}} 

\newcommand{\Sy}{\mathcal{S}}

\newcommand{\R}{\mathcal{R}}
\newcommand*\colvec[3][]{
	\begin{pmatrix}\ifx\relax#1\relax\else#1\\\fi#2\\#3\end{pmatrix}
}
\renewcommand{\S}{\mathcal{S}}
\newcommand{\T}{\mathcal{T}}

\newcommand{\E}{\mathsf{E}}
\newcommand{\F}{\mathsf{F}}
\newcommand{\G}{\mathsf{G}}
\renewcommand{\R}{\mathcal{R}}

\newcommand{\hisr}{\mathcal{H}_{\mathcal{S}} \otimes \mathcal{H}_{\mathcal{R}}}

\usepackage{physics}
\usepackage{soul}
\usepackage{relsize}
\usepackage{slantsc}
\usepackage{bbm}
\usepackage{graphicx}
\usepackage{amsmath}
\usepackage{bbold}
\usepackage{amsfonts}
\usepackage{amssymb}
\usepackage{amsthm}
\usepackage{amsbsy}
\usepackage{mathrsfs}
\usepackage{varioref}
\usepackage{dsfont}
\usepackage{bm}
\usepackage{color}
\usepackage[dvipsnames]{xcolor}
\definecolor{myblue}{rgb}{0.2,0.2,0.8}
\definecolor{myblack}{rgb}{0,0,0}
\definecolor{myurl}{rgb}{0.1,0.1,0.4}
\usepackage[colorlinks=true,citecolor=myblue,linkcolor=myblack,urlcolor=myurl]{hyperref}
\usepackage{cleveref}
\usepackage{subfigure}
\usepackage{booktabs} 
\usepackage{array} 
\usepackage{paralist} 
\usepackage{verbatim} 
\edef\restoreparindent{\parindent=\the\parindent\relax}
\usepackage[parfill]{parskip}
\usepackage{tikz-cd}

\renewcommand\thesection{\arabic{section}}
\renewcommand\thesubsection{\thesection.\arabic{subsection}}

\usepackage{color}
\newtheorem{theorem}{Theorem}[section] 

\newtheorem{proposition}[theorem]{Proposition}
\newtheorem{definition}[theorem]{Definition}
\newtheorem{example}[theorem]{Example}

\newtheorem{remark}[theorem]{Remark}

\usepackage{amsmath}
\usepackage[makeroom]{cancel}

\begin{document}
	
	\title{Operational Quantum Reference Frame\newline Transformations}
	
	\author{Titouan Carette}\email{titouandotcaretteatpolytechniquedotedu}
	\affiliation{Basic Research Community for Physics, Leipzig, Germany}
	\affiliation{LIX, CNRS, {\'E}cole polytechnique, Institut Polytechnique de Paris, 91120 Palaiseau, France}
	\author{Jan G{\l}owacki}\email{jan.glowacki.research@gmail.com}
	\affiliation{Basic Research Community for Physics, Leipzig, Germany}
	\affiliation{Center for Theoretical Physics, Polish Academy of Sciences, 02-668 Warsaw, Poland}
	\affiliation{International Center for Theory of Quantum Technologies, University of Gda{\'n}sk, 80-309 Gda{\'n}sk, Poland}
	\affiliation{Department of Computer Science, University of Oxford, OX1 3QD Oxford, UK}
	\affiliation{Institute for Quantum Optics and Quantum Information, Austrian Academy of Sciences, Boltzmanngasse 3, A-1090 Vienna, Austria}
	\author{Leon Loveridge}\email{leon.d.loveridge@usn.no}
	\affiliation{Basic Research Community for Physics, Leipzig, Germany}
	\affiliation{Department of Science and Industry Systems, University of South-Eastern Norway, 3616 Kongsberg, Norway}
	\affiliation{Okinawa Institute of Science and Technology Graduate University,
		1919-1 Tancha, Onna-son, Okinawa 904-0495, Japan}

	
	\begin{abstract}
		Quantum reference frames are needed in quantum theory for many of the same reasons that reference frames are in classical theories: to manifest invariance in line with fundamental relativity principles and to provide a basis for the definition of observable quantities. Though around since the 1960s, and used in a wide range of applications, only recently have the means for transforming descriptions between different quantum reference frames been tackled in detail. In this work, we provide a general, operationally motivated framework for quantum reference 
		frames and their transformations, holding for locally compact groups. The work is built around the notion of operational equivalence, in which quantum states that cannot be physically distinguished are identified. For example, we describe the collection of relative observables as a subspace of the algebra of invariants on the composite of system and frame, and from here the set of relative states is constructed through the identification of states which cannot be distinguished by relative observables. Through the notion of framed observables---the formation of joint observables of system and frame---of which the relative observables can be understood as examples, quantum reference frame transformations are then  maps between equivalence classes of relative states which respect the framing. We give an explicit realisation in the setting that the initial frame admits a highly localized state with respect to the frame observable.  The transformations are invertible exactly when the final frame also has such a localizability property. The procedure we present is in operational agreement with other recent inequivalent constructions on the domain of common applicability, but extends them in a number of ways, and weakens claims of entanglement generation through frame changes.
	\end{abstract}

	\maketitle

	\section{Introduction}
	
	A means for describing one quantum system relative to another is needed to accommodate the imposition of a symmetry principle. For example, in the absence of absolute space, and in order to respect Galilean or special relativistic invariance, the position of a quantum system has meaning only in relation to other quantum systems, which then function as \emph{quantum reference frames}. A full account of such relative descriptions, at the level of states, observables, and the concomitant 
	probability measures, and how these descriptions transform when different frames are chosen, is the subject of this paper. Our contribution is inspired by other recent efforts with the same aim  (e.g. \cite{giacomini2019quantum,de2020quantum,de2021perspective}), with the further ambitions of (i) providing a mathematically precise framework for locally compact groups, and (ii) making the framework as far as possible \emph{operational} \cite{busch1997operational}, in the sense of being grounded ultimately in Born rule probabilities that could in principle be observed. In so doing, we provide a foundation for the study of quantum reference frames rooted in the full probabilistic structure of quantum theory (observables as positive operator-valued measures, states as positive trace class operators with unit trace) with a symmetry principle explicitly incorporated from the outset. Since there are a number of differing contemporary frameworks and perspectives on the founding assumptions and mathematical implementations of quantum reference frames and their transformations (e.g. \cite{Bartlett2007,palmer2014changing,Loveridge2017a,Vanrietvelde:2018pgb}), with differing takes on fundamental questions such as whether frame changes are unitary \cite{castro2021relative,lake2023quantum}, the desire for a precise formulation is clear. 
	
	The need to understand quantum reference frames and their transformations arises in diverse areas of physics, including curved spacetime quantum field theories and quantum gravity, in which background independence suggests working with relational observables (e.g., \cite{chandrasekaran2023algebra,witten2024background,fewster2024quantum}), quantum information theory (e.g. \cite{Bartlett2007,palmer2014changing}) where communication tasks may require transmission of frame-change-invariant information, and
	in the foundations of quantum mechanics where an understanding of the basic meaning of the formalism is still being sought.\footnote{In this vein, one can include  \emph{relational quantum mechanics}, originally due to Rovelli \cite{rovelli1996relational,adlam2022information}, and the perspectival approach of Bene and Dieks \cite{bene2002perspectival}. These are not based on quantum reference frames explicitly, but the motivation is largely aligned.}
	This assortment of applications
	explains the variety of distinct formulations---perhaps four strands are now visible, each described by one of \cite{Bartlett2007,giacomini2019quantum,Vanrietvelde:2018pgb,Loveridge2017a}---with differing underlying assumptions. In quantum information theory, the frame-independent communicable information between agents  who do not share a frame is achieved through averaging the states over the group of all (classical) frame orientations. By contrast, the perspective neutral approach \cite{Vanrietvelde:2018pgb,de2021perspective}, inspired by Dirac quantization of constrained systems, limits the physical (pure) states to the unit vectors in the given Hilbert space that are invariant under the action of the group generated by the quantized constraints, which yields different physics than that given in the quantum information approach. Another major recent development is given in  
	\cite{giacomini2019quantum}, in which `states relative to a frame' for a particle on a line are taken as primitive, and explores the ensuing change of state description under frame changes, resulting in e.g. the generation of entanglement. This direction was taken up in \cite{de2020quantum}, in which the group theoretic aspects are emphasized, various assumptions are made more explicit, and a more general treatment is given. In each approach, akin to special relativity and gauge theories, there is a group which dictates both what is observable/physical in the theory and also mediates transformations between frames, and the distinction between the approaches can in part be understood through the role played by the symmetry group. In this paper, we take that the physical \emph{observables} are invariant to be a fundamental physical principle, and explore the consequences, culminating in a frame change procedure based on ideas arising in e.g. \cite{loveridge2012quantum,loveridge2017relativity,Loveridge2017a}. We comment on relative merits, drawbacks and distinctions between the various approaches as we go.

	\subsection{Survey}
	
	The paper is organised as follows. Sec. \ref{sec:pre} begins with the provision of some general background and nomenclature used later,
	focussing on states, observables, channels and probabilities, 
	and providing some general topological material. The state spaces we consider are more general than those usually encountered in quantum mechanics, understood as convex subsets of real vector spaces. Next, we 
	formulate the notion of \emph{operational equivalence}, based on identifying (typically) states that cannot be probabilistically distinguished, either for principle or practical reasons. Concretely, the identification (on the trace class) is given as $T \sim_\mathcal{O} T '$ if $\tr[T A] = \tr[T' A]$ for all $A \in \mathcal{O} \subset B(\hi)$, where $\mathcal{O}$ is some specified collection of operators; the ensuing quotient on density operators is then the generic model for the state spaces we will encounter throughout. Some general theorems about such spaces are given; most notably we show that the standard statistical duality between the states and effects can be given in general operational terms under the above quotient. 
	
	After these generalities, we consider the example of operational equivalence that arises under a symmetry described by a unitary representation of a locally compact group $G$, and with the set $\mathcal{O}$ above taken to be the fixed point algebra $B(\hi)^G$ of $G$-invariant bounded operators. The state space then comprises the equivalence classes of density operators that cannot be distinguished by elements of $B(\hi)^G$. This setting is discussed in some detail, since it serves as an illustration of the conceptual and formal differences between the approach to implementing invariance employed in this work and in other approaches to quantum reference frames.

	We choose the invariant operators rather than invariant states (as in the quantum information approach) or Hilbert space vectors (as in the perspective-neutral approach \cite{de2021perspective}),
	because when $G$ is not compact such objects may not exist. Indeed, the perspective-neutral approach requires distributional techniques which are not known to be rigorously possible in general. Further operational equivalences are considered later; 
	the newly introduced notions of relative states and observables, 
	needed for the frame transformations, are based on this concept. 
	
	Next comes the topic of the \emph{localizability} (norm-$1$) property of positive operator-valued measures (POVMs), which assures the existence of probability measures arising via the Born rule being concentrated around an arbitrary point to arbitrary precision. These POVMs, in conjunction with a 
	\emph{covariance} property described here, play an important role throughout, as discussed in the next section.
	Various examples of covariant POVMs are provided (with and without the localizability property being satisfied), including those defined by systems of coherent states, familiar from the perspective-neutral approach (e.g. \cite{Vanrietvelde:2018pgb,de2021perspective}).

	Sec. \ref{sec:rqk} provides a definition of a quantum reference frame $\R$ as a \emph{system of covariance} $\R= ( U_\R,\E_\R,\hir)$ based on $\Sigma_{\R}$, i.e, a unitary representation of $G$ in a Hilbert space $\hir$ and a covariant POVM $\E_\R$ on a measurable space $\Sigma_{\R}$ with values in $B(\hir)$. Here, the compound object of system $\S$ and frame $\R$ is considered for the first time.
	Frames are understood as (representing) physical systems, and fixing a frame is viewed as setting a concrete experimental arrangement,
	akin to fixing a measuring apparatus and choosing a particular pointer observable in the quantum theory of measurement (e.g., \cite{Busch1996}). Thus, the choice of frame reflects which system will execute that function, and how it will do it. From here we introduce \emph{framing}, which yields an operational equivalence relation reflecting the fixing of the frame and therefore respecting the given operational/experimental scenario under consideration.
	
	In this work, we are primarily concerned with the setting of \emph{principal} frames, in which the value space $\Sigma_{\R}$
	is a principal homogeneous $G$-space, which already covers a large range of physically relevant examples. In this context, the frame observables are understood as providing a means of describing, probabilistically, the \emph{orientation} of the quantum system, contingent on the state preparation of the frame; using localizing sequences of states we make rigorous---in terms of a limiting procedure at the level of probabilities/expectation values---objects such as $\ket{g}$ ($g \in G$) which are used freely in the physics literature but require proper mathematical foundation. We note that the orientation is not {\it per se} observable, since we stipulate that truly observable quantities are invariant, and we will define relative orientation observables which do respect the invariance requirement.
	
	Next we recall the \emph{relativization maps} $\Y^\R$ on principal homogeneous spaces, introduced in e.g. \cite{Miyadera2015e,Loveridge2017a}, which are at the heart of the operational approach to quantum reference frames.  The setting of general homogeneous spaces is considered for finite groups/spaces in \cite{glowacki2023quantum}, and \cite{fewster2024quantum} deals with the case of $G/H$ for $G$ locally compact and $H$ the compact stabiliser of the action. A relativization map $\Y^\R$ is a quantum channel, and takes $B(\his)$ to invariants on the composite system, i.e., $\Y^\R(B(\his)) \subset B(\hisr)^G$. The (ultraweak closures of the) images of the relativization maps provide what we call spaces of \emph{relative observables}, written $B(\his)^\R := \Y^\R(B(\his))^{\rm cl}$. Such operators are not only invariant but also framed, and understood as those that can actually be measured, given the imposition of symmetry (invariance) and the choice of frame. We note that the collection of relative observables need not form an algebra, but is always an operator space.
	
	Following the definition of relative observables, the \emph{relative states} $\S(\his)_\R$ (and relative trace class operators $\T(\his)_\R$) are defined as operational equivalence classes of states (trace class operators) on the composite system that cannot be distinguished by relative observables. The space of relative observables $B(\his)^\R$ is in Banach duality with the space of relative trace class operators---$[\T(\his)_\R]^* \cong B(\his)^\R \subset B(\hisr)^\G$---with the relative states  embedded as $\S(\his)_\R \subseteq \T(\his)_\R^{\rm sa}$ (`sa' standing for `self-adjoint'); this extends the standard duality $\T(\hi)^* \cong B(\hi)$  between the trace class and bounded operators to their relational incarnation considered in this work. For the comparison with other attempts to formalise the notion of relative state, we note that the relative state spaces introduced here are isomorphic to convex subsets of the state space of the system $\S$, namely to the image of the predual map $\Y^\R_*$ to be written $\S(\hi)^\R := \Y^\R_*(\S(\hisr))\cong \S(\hi)_\R \subseteq \S(\his)$. This allows for the representation of the relative states as states of the system alone, which is in line with what is done in other works (e.g. \cite{giacomini2019quantum,de2020quantum}), except more general, operationally motivated and mathematically precise. 
	
	We then consider conditioned relative descriptions, which arise from fixing a state of the given frame.  The main tool here is the assignment to any frame state $\omega \in \S(\hir)$ the corresponding \emph{restriction map} $\Gamma_\omega: B(\hisr) \to B(\his)$, which  is known to be a completely positive normal conditional expectation \cite{Miyadera2015e,Loveridge2017a}. This naturally gives rise to the dual notions of the \emph{conditioned relative observables} and the \emph{product relative states}, which live in the image of $\Y^{\R}_*$ applied to product states. These may be used to generalise the alignable states of \cite{krumm2021quantum}; various properties of this particularly tractable class of relative states are given. Given the means for  conditioning upon a state of the reference, we address the question of how the conditioned relative descriptions behave upon localizing the reference. Generalizing previous results (of e.g. \cite{loveridge2012quantum,Loveridge2017a}), we find that in the case of a localizable frame, any observable of the system can be represented as a limit of conditioned relative observables,   $\Y^\R$ is injective, and the set of relative states is dense in the state space of the system. Thus the ordinary, non-relational framework of quantum theory is recovered as a limiting case of the description relative to a good quantum reference frame.
	
	The work thus far described may then be used in service of the operational take on frame changes, to which the paper turns in Sec. \ref{sec:fcp}. The transformation rule we give is inspired by
	\cite{giacomini2019quantum,de2020quantum}, but mathematically and physically distinct from it, and all other quantum frame changes appearing to date. The map we provide is defined on the relevant convex operational relative state space. Here, three systems are considered - a system and two frames, and the aim is to transform the (state) description relative to one frame to that relative to the other. Before giving the frame change, we introduce a final ingredient called the \emph{lifting maps}, which serve to `attach' a chosen state of the second reference to a state relative to the first one, giving rise to a state on the total system consisting of both frames and the original system, and respecting symmetry. With all the tools in hand, we may then define the localized frame transformations as a limiting case of the following procedure. First, given a state relative to the first frame, use the lifting map to attach a state of the first frame that is well-localized at the identity of $G$, in analogy to \cite{de2020quantum}. Thus we pass by the `global' description comprising system and all frames, whilst respecting the symmetry, similar in spirit to what is done in \cite{de2021perspective}. Then we use the predual of the second relativization map, yielding a state relative to the second frame. The limit is taken with respect to an arbitrary localizing sequence centered around the identity of $G$. The procedure is well-defined and invertible when the second frame is also localizable, and both framing equivalence relations are taken into account. We also show that in the setup of three frames, the localized frame transformations compose as one would hope.
	
	We conclude by showing that our map agrees in spirit with
	\cite{giacomini2019quantum,de2020quantum} in the following sense. In the setting that the frame changes in the above works are rigorously defined, those frame changes yield elements of the operational equivalence classes appearing in our work. The upshot is that if one accepts the premise of the operational equivalence relations defined here, there is no measurement that can distinguish our frame change from theirs. Given that the operational equivalence classes contain both entangled states and separable states with the same statistics on relative observables, in our prescription there is no fact of the matter about whether the post-frame change states are entangled - indeed, since no measurement could show that they are, we urge some degree of caution in making strong physical claims based on the generation of entanglement from frame changes. We also provide a brief comparison with the perspective-neutral approach, though a full account of it is a large project due to the difficulties of making the perspective-neutral framework precise. However, we again observe  no operational disagreement between our maps and those of \cite{de2021perspective} in the setting that the initial frame is ideal. After a brief conclusion, in Appendix \ref{glossary} we provide a glossary of the various spaces we consider for reference. Appendix \ref{diagrams} contains all the diagrams in one place for ease of use, and finally Appendix \ref{app:proofs} provides proofs of some theorems omitted from the main text.
	
	\section{Preliminaries}\label{sec:pre}
	
	We briefly recall in this section the essentials of the Hilbert space framework of quantum mechanics pertinent to the subject matter which follows. We introduce the notion of operational equivalence, which is used to identify 
	mathematically distinct entities (most commonly for us these are states) 
	that cannot be distinguished in a given operational scenario, paying specific attention to the equivalence arising from the imposition of a symmetry. We then define the relative states in an operational way, given as the equivalence classes of density matrices that cannot be probabilistically told apart by relative observables. We close by recalling the definitions of covariance and localizability for POVMs, accompanied by a list of examples.
	
	\subsection{Basics}
	
	In the  Hilbert space framework of quantum theory, a quantum system has associated with it a (complex, separable) Hilbert space $\hi$, and (normal) states are identified with positive trace class operators on $\hi$ of trace one; we write $\T(\hi)$ for the trace class (of $\hi$), which is a Banach space under the trace norm $||\cdot ||_1$, $\T(\his)^{\rm sa}$ for the self-adjoint part and $\mathcal{S}(\hi)$ for the convex subset of states. The pure states are the extremal elements of $\mathcal{S}(\hi)$, characterised as the rank-1 projections and written $\dyad{\varphi}$ for some unit vector $\varphi \in \hi$; the collection of pure states is written $\mathcal{P}(\hi)$. The space (which is also a von Neumann algebra) of bounded operators in $\hi$ is denoted $B(\mathcal{H})$, which is the Banach dual of $\T(\hi)$. A channel $\Lambda:B(\hi) \to B(\mathcal{K})$ is a normal (i.e., continuous with respect to the ultraweak topologies on $B(\hi)$ and $B(\mathcal{K})$) completely positive (CP) map which preserves the unit; the normality and unitality imply there is a unique CP trace-preserving \emph{predual} map (also to be called a channel) $\Lambda_*:\T(\hi) \to \T(\mathcal{K})$ defined through
	$\tr[T\Lambda(A)]=\tr[\Lambda_*(T)A]$ for all $T \in \T(\hi)$ and 
	$A \in B(\hi)$, which also takes states to states. More generally we consider abstract \emph{state spaces}, which are convex subsets of real vector spaces, and convexity-preserving maps between them. The convexity-preserving bijections are called state space isomorphisms. Usually, the state spaces considered are total convex subsets of real Banach spaces, where total convexity means that the points are separated by the (real) bounded affine functionals; this is the case for $\S(\hi) \subset \T(\hi)^{\rm sa}$, where the relevant functionals are obtained by tracing with self-adjoint operators. 
	
	Observables are (identified with) positive operator-valued measures (POVMs) $\E: \mathcal{F} \to B(\hi)$, where $\mathcal{F}$ is a $\sigma$-algebra of subsets of some set (or \emph{value space}) $\Sigma$, which represents outcomes that may be obtained in a measurement of the given observable. 
	In this paper, $\Sigma$ is a topological space and $\mathcal{F}$ is the Borel $\sigma$-algebra, written $\mathcal{B}(\Sigma)$.
	If the system is prepared in a state $\omega \in \S(\hi)$, upon measurement of $\E$ the probability that an outcome is in a set 
	$X \in \mathcal{B}(\Sigma)$ is given 
	through the Born formula:
	\begin{equation}\label{eq:born}
		p_{\omega}^{\E}(X) = \tr[\omega \, \E (X)].
	\end{equation}
	Operators in the range of POVMs, that is, positive operators in the unit operator interval, are called \emph{effects}. Thanks to the duality between the bounded and trace class operators, the effects can be equivalently characterized as positive continuous affine functionals on $\S(\hi)$ bounded by one. The space of effects on $\hi$ will be denoted by $\Eff(\hi)$. If each operator in the range of a POVM $\E$ is a projection, $\E$ is a projection-valued measure (PVM) and if defined on (Borel subsets of) $\mathbb{R}$ the standard description of observables as self-adjoint operators in $\hi$ is recovered through the spectral theorem. POVMs which are PVMs will be called \emph{sharp} (observables), all others \emph{unsharp}. In the case that
	an observable is sharp, we will also refer to the corresponding self-adjoint operator as an observable. Note that `observable algebras' typically have non-self-adjoint elements, as one finds in the algebraic QM/QFT literature.
	
	Various topological notions have already been mentioned; it is worth briefly formalizing these since they appear repeatedly. Being motivated by operational
	ideas, the preferred topology on $B(\hi)$ is the \emph{topology of pointwise convergence of expectation values}, i.e., $A_n \to A \in B(\hi)$ exactly when $\tr[T A_n] \to \tr[T A]$ for all $T \in \T(\hi)$ (there is no loss of generality to restrict further from $\T(\hi)$ to $\mathcal{S}(\hi)$). This is the ultraweak (also called $\sigma$-weak, or weak-$*$ as the dual of $\T(\hi)$) topology on $B(\hi)$ (on norm bounded sets convergence here agrees with convergence in the weak operator topology arising from the family of seminorms $A \mapsto |\ip{\varphi}{A \phi}|$, which can be reconstructed from pure state expectation values by polarization). Note that all these topologies are Hausdorff \cite{murphy2014c}. On the predual $\T(\hi)$ (and by restriction to the state space), 
	we again use the topology of pointwise convergence of the expectation values, thus $T_n \to T$ exactly when $\tr[T_n A] \to \tr[T A]$ for all $A \in B(\hi)$. Since the set of effects spans $B(\hi)$, this can be equivalently given in $\mathcal{E}(\hi)$. We will refer to this topology as the \emph{operational topology} on the trace class (and accordingly also on the states) i.e., it is the weakest topology that makes all the $T \mapsto \tr[T A]$ continuous. The superscript "$^{\rm cl}$" will always refer to ultraweak closure of the subsets in operator algebras, and operational closure of subsets of trace class operators.
	
	\subsection{Operational Equivalence}
	
	The notion of operational equivalence allows for the identification of distinct `entities' that cannot be distinguished physically/operationally. 
	For instance, in the operational approach to quantum theory (e.g. \cite{busch1997operational}), states are defined as equivalence classes of preparation procedures that cannot be distinguished in any measurement. In gauge theories, fields related by a gauge transformation may be regarded as physically identical. In this work we specify a set of observables, and quotient the set of states by an operational equivalence relation defined by equality of expectation values, i.e., we identify states that cannot be probabilistically distinguished on the given set of observables. The prototypical example comes from identifying any states which give the same statistics on all (gauge-)invariant observables. The state spaces which arise in this way are naturally identified with the predual of the von Neumann algebra of invariant observables $B(\his)^G$. This operational state space  is a central object in this work, and plays a role similar to that of the physical Hilbert space in the perspective-neutral approach \cite{Vanrietvelde:2018pgb,de2021perspective}. We also consider an important case of (operationally) equivalent collections of density operators that cannot be distinguished by operators in the image of a given channel, which will come useful in the context of the construction of relative states.
	
	\subsubsection{Generalities}
	
	We begin with a general definition of operational equivalence for collections of states and effects:
	
	\begin{definition}
		Let $\mathcal{S}$ be a collection of states and $\mathcal{O}$ a collection of effects. Then:
		\begin{itemize}
			\item Two states $\rho, \rho' \in \mathcal{S}$ are called \emph{operationally equivalent} with respect to $\mathcal{O}$
			if $\tr[\rho \, \F]=\tr[\rho' \, \F]$ for all $\F \in \mathcal{O}$.
			\item Two effects $\F,\F' \in \mathcal{O}$  are called operationally equivalent with respect to $\mathcal{S}$ if $\tr[\rho\F]=\tr[\rho \F']$ for all $\rho \in \mathcal{S}$.
		\end{itemize}
	\end{definition}
	This definition may be adapted to POVMs, self-adjoint and trace class operators as needed. Since the following will be used throughout, we state it separately: 
	
	\begin{definition}
		Given a subset $\mathcal{O}\subseteq B(\mathcal{H})$, the $\mathcal{O}$-\emph{operational equivalence relation} on $\T(\mathcal{H})$ is defined as
		\[T \sim_{\mathcal{O}} T' \Leftrightarrow \tr[TA]=\tr[T'A] \hspace{3pt} \forall A \in \mathcal{O}.\]
	\end{definition}
	The identification of $\mathcal{O}$-equivalent trace class operators amounts to taking the quotient $\T(\mathcal{H})/\hspace{-3pt}\sim_{\mathcal{O}}$.
	The continuity of the map $A \mapsto \tr[T A]$ on $B(\hi)$ for any $T \in \T(\hi)$ means that the set of $\mathcal{O}$-equivalent trace class operators is equal to the set of $\mathcal{O}^{\rm{cl}}$-equivalent trace class operators (e.g. \cite{kelley2017general}). Notice also that the same equivalence is realized with respect to $\mathcal{O}$, ${\rm span}\{\mathcal{O}\}$ and ${\rm conv}\{\mathcal{O}\}$ (denoting the convex hull); these facts will be used throughout. 
	
	\begin{proposition}\label{prop:iml}
		The space $\T(\mathcal{H})/\hspace{-3pt}\sim_{\mathcal{O}}$ equipped with the quotient norm is the unique Banach predual of the ultraweak closure of the span of $\mathcal{O}$, so that
		\[
		\big[{\rm span}\{\mathcal{O}\}^{\rm cl}\big]_* \cong \T(\mathcal{H})/\hspace{-3pt}\sim_{\mathcal{O}}.
		\]
		
	\end{proposition}
	\begin{proof}
		For $A \in \mathcal{O}$, we write $\phi_A$ for the (trace norm) continuous linear functional $T \mapsto \tr[A T]$ and identify $A$ with $\phi_A$ through the
		isomorphism $\T(\mathcal{H})^* \cong B(\mathcal{H})$. Then
		\[
		T \sim_{\mathcal{O}} T' \Longleftrightarrow \hspace{3pt} T-T' \in \ker(\phi_A) \hspace{3pt} \forall A \in \mathcal{O}.
		\]
		
		Notice now that the pre-annihilator of $\mathcal{O}$, defined as the joint kernel of the functionals $\phi_A$ with $A \in \mathcal{O}$ and denoted by
		\[
		{}^{\perp}\mathcal{O} := \bigcap_{A\in \mathcal{O}}\ker(\phi_A) \subseteq \T(\hi),
		\]
		is closed as an intersection of closed sets, and a subspace due to linearity of the $\phi_A$ functionals. Thus ${}^{\perp}\mathcal{O} = {}^{\perp}{\rm span}\{\mathcal{O}\}$ and $\T(\hi) /\hspace{-3pt}\sim_{\mathcal{O}} = \T(\hi) /{}^{\perp}\mathcal{O}$ is a Banach space with the quotient norm given by
		\[
		||T + {}^{\perp}\mathcal{O} || :=\inf_{S \in {}^{\perp}\mathcal{O}} ||T+S||_1.
		\]
		
		We then have $\left(\T(\hi)/\hspace{-3pt}\sim_{\mathcal{O}} \right)^* = \left(\T(\hi)/{}^{\perp}\mathcal{O}\right)^* \simeq {}^{\perp}\mathcal{O}^{\perp} = \left({}^{\perp}{\rm span}\{\mathcal{O}\}\right)^{\perp} = {\rm span}\{\mathcal{O}\}^{\rm cl}$ (see Thm. 4.9 and 4.7 in \cite{rudin1974functional}).
		
		For the uniqueness, we note the following. Since the set of compact operators $K(\hi)$ is convex, the ultraweak closure $K(\hi)^{\rm cl}$ equals the $\sigma$-strong$^*$ closure of $K(\hi)$ (Thm. 2.6 (iv) in \cite{takesaki2001}). Due to the $\sigma$-strong$^*$ continuity of the $^*$-operation, $K(\hi)^{\rm cl}$ is a $\sigma$-strongly$^*$ closed $*$-subalgebra of $B(\hi)$ and the double commutant theorem gives $K(\hi)^{\rm cl}=K(\hi)''=B(\hi)$. Then clearly $K(\hi)\cap {\rm span}\{\mathcal{O}\}^{\rm cl}$ is ultraweakly dense in ${\rm span}\{\mathcal{O}\}^{\rm cl}$. Since we assume separability of $\hi$, the uniqueness then follows from Thm. 2.1 in \cite{godefroy2014uniqueness}.

	\end{proof}
	\begin{remark} The uniqueness of preduals is known  in the theory of von Neumann algebras \cite{Sakai1971-wd}, extended here to ultraweakly closed operator spaces (which we recall are norm-closed subspaces of Banach spaces). Indeed, applying the above  reasoning and Thm. 2.1 in \cite{godefroy2014uniqueness} gives uniqueness of Banach preduals for arbitrary ultraweakly closed subspaces of $B(\hi)$ for separable $\hi$. These predual spaces are unique up to isometry and can be equally well characterized as the spaces of normal functionals on such subspaces (see 2.6 (iii) in \cite{takesaki2001} and  \cite{zhong-jin1992}). The uniqueness of the predual spaces allows for the unambiguous switching between dual descriptions (\textit{\`a la} Schr\"{o}dinger/Heisenberg pictures), which we apply in various operational settings later.
	\end{remark}

	The set of classes of indistinguishable density operators can be understood as a state space in its own right, as is confirmed by the following.

	\begin{proposition}\label{prop:statespace}
		The set $\S(\mathcal{H})/\hspace{-3pt}\sim_{\mathcal{O}}$ is a total convex state space as a subset of the real Banach space $\T(\hi)^{\rm{sa}}/\hspace{-3pt}\sim_{\mathcal{O}}$. Moreover, it is closed in the quotient operational topology.
	\end{proposition}
	
	\begin{proof}
		Since the real linear structure of $\T(\hi)^{\rm{sa}}/\hspace{-3pt}\sim_{\mathcal{O}}$ comes from $\T(\hi)^{\rm{sa}}$, convexity is preserved under the quotient. In particular, writing $[\_]_\mathcal{O}$ for the $\mathcal{O}$-equivalence classes, for any $\rho,\rho' \in \S(\hi)$ and $0 \leq \lambda \leq 1$ we have
		\[
		\lambda[\rho]_\mathcal{O} + (1-\lambda)[\rho']_\mathcal{O} = [\lambda \rho + (1-\lambda)\rho']_\mathcal{O} \in \S(\hi)/\hspace{-3pt}\sim_{\mathcal{O}}.
		\]
		
		The bounded affine functionals $\S(\mathcal{H}) \to [0,1]$ are given by $\rho \mapsto \tr[\rho \, \F]$ for with $\F \in \Eff(\hi)$. The effects on $\S(\mathcal{H})/\hspace{-3pt}\sim_{\mathcal{O}}$, i.e, bounded affine functionals $\S(\mathcal{H})/\hspace{-3pt}\sim_{\mathcal{O}} \to [0,1]$, are then given by $\rho \mapsto \tr[\rho \, \F]$ with $\F \in \mathcal{E}(\hi)$ being well-defined on the $\mathcal{O}$-equivalence classes, i.e., such that for all $\rho, \rho' \in \S(\hi)$ we have
		\[
		\rho \sim_\mathcal{O} \rho' \Rightarrow \tr[\rho \, \F] = \tr[\rho' \, \F].
		\] 
		Since the operational equivalence classes taken with respect to $\mathcal{O}$ and $\rm{span}\{\mathcal{O}\}^{\rm cl}$ are the same, and clearly the effects outside of $\rm{span}\{\mathcal{O}\}^{\rm cl}$ will not be well-defined on the $\mathcal{O}$-equivalence classes, we can conclude that $\F \in \rm{span}\{\mathcal{O}\}^{\rm cl}$. The effects on $\S(\mathcal{H})/\hspace{-3pt}\sim_{\mathcal{O}}$ are then given by the operators in $\Eff(\hi) \cap \rm{span}\{\mathcal{O}\}^{\rm cl}$. Such effects separate the elements of $\S(\mathcal{H})/\hspace{-3pt}\sim_{\mathcal{O}}$ by construction, providing total convexity.
		
		The state space $\S(\hi)$ is operationally closed in $\T(\hi)$, since for any sequence of states $(\rho_n) \subset \S(\hi)$ such that $\lim_{n \to \infty}\tr[\rho_n \, A] = \tr[T \, A]$ for all $A \in B(\hi)$ and some $T\in \T(\hi)$, we can conclude that $T \in \S(\hi)$. Indeed, the continuity of the trace gives positivity and normalization of $T$. The operational topology on $\S(\hi)/\hspace{-3pt}\sim_{\mathcal{O}}$ is the quotient topology of the one on $\T(\hi)$ so we have
		$$\lim_{n \to \infty} [\rho_n]_\mathcal{O} = [T]_\mathcal{O} \in \S(\hi)/\hspace{-3pt}\sim_{\mathcal{O}}.$$
	\end{proof}

	A state space of the form above will be called an \emph{operational state space}. Total convex subsets of real Banach spaces can be embedded in the general framework of base-norm spaces and the dual order unit spaces (e.g., \cite{beltrametti1997effect}), pointing to potential generalizations of the notions presented in this work.
	
	Often in this paper, the set $\mathcal{O}$ is the image of a unital normal positive map. In this case, the corresponding state space admits an equivalent useful characterization. 
	
	\begin{proposition}\label{generalst}
		Any normal, positive, unital map $\Lambda:B(\mathcal{K})\to B(\mathcal{H})$ provides a state space isomorphism 
		\[
		\S(\mathcal{H})/\hspace{-3pt}\sim_{\Im \Lambda}
		\cong \Lambda_* (\S(\mathcal{H})),
		\]
	\end{proposition}
	
	\begin{proof}
		Since $\Lambda$ is normal, we can write ${}^\perp \Im \Lambda = \ker \Lambda_* $, and thus $\T (\mathcal{H})/\hspace{-3pt}\sim_{\Im \Lambda} = \T (\mathcal{H})/\ker \Lambda_*$. Then $\Lambda_* $ restricts to an \emph{invertible} bounded linear map \newline $\T (\mathcal{H})/\hspace{-3pt}\sim_{\Im \Lambda} \to \Im \Lambda_* $. Since $\Lambda$ is linear, unital and positive, $\Lambda_*$ restricts further to a convexity preserving bijection $\S (\mathcal{H})/\hspace{-3pt}\sim_{\Im \Lambda} \to \Lambda_* (\S(\mathcal{H}))$.\footnote{Note that in general, this correspondence doesn't hold at the level of the ambient Banach spaces since $\Im \Lambda_* $ might not be closed.}
	\end{proof}
	
	\subsubsection{Operational equivalence from symmetry}\label{subsubsec:oes}

	An important example of operational equivalence arises in the presence of symmetry. Consider a strongly continuous unitary representation $U:G \to B(\hi)$ of a (locally compact) group $G$, writing $g.A$ to stand for $ U(g) A U(g)^* $ with $A\in B(\hi)$ and $T.g$ for $U(g)^* T U(g)$ with $T \in \mathcal{T}(\hi)$.
	Two density operators are identified if they cannot be distinguished by observables invariant under the representation, and therefore \emph{invariant algebra} $B(\hi)^G := \{A \in B(\hi)~|~g.A = A\}$ (which is a von Neumann algebra as the commutant of a unitary group) plays a key role. Moreover, as we shall see, observables with effects in the invariant algebra can be understood as those which can be defined without reference to an external frame.
	
	\begin{definition}
		The operational equivalence relation on $\T(\hi)$ taken with respect to the invariant algebra $B(\his)^G$ is denoted $T \sim_G T'$. The space of $G$-\emph{equivalent trace class operators} is given by
		\[
		\T(\hi)_G := \T(\hi)/\hspace{-3pt}\sim_{G},
		\]
		while the operational state space of $G$-\emph{equivalent states} is defined to be
		\[
		\S(\hi)_G := \S(\hi)/\hspace{-3pt}\sim_G.
		\]
	\end{definition}
	Prop. \ref{prop:statespace} and \ref{prop:iml} then ensure the following:
	\begin{proposition}
		There is a Banach space isomorphism
		\[
		B(\his)^G_* \cong \T(\his)_G.
		\]
		Moreover, the subset $\S(\his)_G \subset \T(\his)^{sa}_G$ is a total convex state space. 
	\end{proposition}

	Note that in general $\S(\hi)_G$ cannot be identified with the set $\S(\hi)^G $ of invariant states, i.e., those $\rho\in \S(\hi)$ for which $\rho.g =\rho $ (here and in the sequel, a superscript $G$ denotes true invariants, and a subscript $G$ the relevant quotient space), this last set being generically empty if $G$ is not compact. By contrast, there is often an abundance of invariant observables, which is most clearly seen in the setting of a composite system. The state space $\S(\hi)_G$ is generically non-trivial and justifies on mathematical grounds that invariance should be stipulated on the observables rather than states (see \cite{Loveridge2017a,Miyadera2015e,loveridge2017relativity} on which the present approach is based). The $G$-equivalent states defined here are similar to the "symmetry-equivalent" states defined in \cite{krumm2021quantum} (Def.18) in the context of finite abelian groups. We note also that there is some similarity with the method of co-invariants, in which states on a $G$-orbit are identified. However, the $G$-equivalence here differs from both invariance and co-invariance, thereby yielding a new class of physical states which have not been studied in full to date. We are also able to avoid the use of distributions/rigged Hilbert spaces which are needed for constructing the physical Hilbert space in the perspective-neutral approach \cite{de2021perspective}, which is defined as the space of invariant (`kinematical') Hilbert space vectors/distributions. We note immediately that by setting $B(\hi)^G$ as the collection of operators on which the equivalence of states is defined, any state on a $G$-orbit is operationally indistinguishable/equivalent to any other. This is as one would expect of gauge transformations: all states related by a gauge transformation are physically equivalent. 
	
	If the group $G$ is compact the $G$-\emph{twirl} (or \emph{incoherent group average}) $\mathcal{G}:B(\hi)\to B(\hi)$ is given by
	\begin{equation}
		\mathcal{G}(A)=\int_G  U(g)A U(g)^*  \, d\mu(g) ,
	\end{equation}
	where $\mu$ is (normalised) Haar measure. $\mathcal{G}$ is a unital normal map with pre-dual $\mathcal{G}_* : \T(\hi) \to \T(\hi)$ taking the form
	$\mathcal{G}_* (\rho)=\int_G U(g)^* \rho U(g)  \, d\mu(g)$. Both $\mathcal{G}$ and $\mathcal{G}_* $ are idempotent, respectively onto the sets $B(\hi)^G$ and $\T(\hi)^G$ of invariant bounded, respectively trace class, operators. The image under $\mathcal{G}_*$ of a state is often interpreted (though not always with clear operational meaning) as an invariant `version' of the given state. If $G$ is not compact the integral does not converge in general on states or operators, and we therefore avoid it in this setting.
	
	{Interestingly, in the case of compact $G$, there is no operational difference between stipulating the invariance requirement on observables or states (or both):
		\begin{proposition}\label{prop:tw1}
			$\tr[\mathcal{G}^*(A) \, \rho]=\tr[A \, \mathcal{G}(\rho)]=\tr[\mathcal{G}^*(A) \, \mathcal{G}(\rho)]$.
		\end{proposition}
		
		The proof uses the invariance of $\mu$ and is straightforward. Moreover, by Prop. \ref{generalst}, for $G$ compact there is an isomorphism of state spaces $S (\mathcal{\hi})^G = \mathcal{G}_*(\S(\hi)) \cong \S (\mathcal{\hi})_G $. Furthermore, in this case, the correspondence lifts to the ambient Banach spaces.
		
		\begin{proposition}
			If $G$ is compact, then the Banach spaces $\T(\hi)^G $ and $\T(\hi)_G $ are isometrically isomorphic.
		\end{proposition}
		
		\begin{proof}
			We have $ \T(\hi)_G = \T(\hi)/\ker(\mathcal{G}_* ) $ and $\T(\hi)^G = \Im \mathcal{G}_* $. Moreover, $\mathcal{G}_*$ factorizes through a bijective map $\tilde{\mathcal{G}}_*:\T(\hi) /\ker(\mathcal{G}_* ) \to \Im \mathcal{G}_* $. We show that $\tilde{\mathcal{G}}_*$ is an isometry. Notice first, that $\mathcal{G}_* $ is a contraction since
			\[
			|| \mathcal{G}_* (T) ||_1 = \no{
				\int_G g\cdot T  \, d\mu(g)}_1\leq \int_G || g\cdot T ||_1  \, d\mu(g) \leq \int_G ||T||_1   \, d\mu(g)= ||T ||_1,
			\]
			and therefore for all $\sigma \in \ker(\mathcal{G}_* )$ we have
			\[
			|| \tilde{\mathcal{G}}_*(T + \ker(\mathcal{G}_* )) ||_1 = || \mathcal{G}_* (T) ||_1= || \mathcal{G}_* (T + \sigma) ||_1 \leq || T + \sigma ||_1.
			\]
			
			Hence $||\tilde{\mathcal{G}}_*(T+ \ker(\mathcal{G}_* )) ||_1\leq ||T + \ker(\mathcal{G}_* )||_1= \inf_{\mu \in \ker(\mathcal{G}_*)} ||T + \mu ||_1$. Then, since $\mathcal{G}_*$ is idempotent we also have
			\begin{align*}
				||T + \ker(\mathcal{G}_* )||_1&= \inf_{\mu \in \ker(\mathcal{G}_*)} ||T + \mu ||_1 \leq ||T + (\mathcal{G}_* (T) - T ) ||_1\\
				 &= ||\mathcal{G}_* (T) ||_1=||\tilde{\mathcal{G}}_*(T+ \ker(\mathcal{G}_* )) ||_1, 
			\end{align*}

			which shows that $\tilde{\mathcal{G}}_*$ is indeed isometric.
		\end{proof}
		
		The setting of invariant observables and operational state spaces therefore aligns with e.g. the quantum information approach to QRFs \cite{Bartlett2007} in the case of compact groups, but is better suited for generalization to the locally compact case.
		
		\subsection{Localizability}
		
		The localizability property of a POVM described below is used to recover the standard kinematics of quantum mechanics from the relational one presented in the next section (see also \cite{heinonen2003norm,Miyadera2015e,Loveridge2017a}), and in defining frame transformations later. Through the Born rule probabilities (eq. \eqref{eq:born}), for a fixed observable $\E$, a state gives rise to a probability measure on $\Sigma$, often to be denoted by $\mu^\E_\omega$. We will refer to states as being `highly localized' with respect to $\E$ in some (open) set $X$ or around some point $x \in \Sigma$ (i.e., in an open neighbourhood $X$ containing $x$), meaning that the given probability $\mu^\E_\omega(X)$ is close to unity. If $\mathsf{P}: \mathcal{B}(\Sigma) \to B(\hi)$ is a PVM, for any $X \in \mathcal{B}(\Sigma)$ for which $\mathsf{P}(X) \neq 0$, there is a state $\omega$ for which $\tr[\omega \, \mathsf{P} (X)] = 1$ (set $\omega$ to be a projection onto any unit vector in the range of $\mathsf{P}(X)$), and this state, with respect to $\mathsf{P}$, is perfectly localized (with probabilistic certainty) in $X$. For example, if $\mathsf{P}=\mathsf{P}^A$, then any eigenvector of $A$ with  eigenvalue in $X$ can be understood in this way. However, since $A$ may have no eigenvalues (e.g. the position operator on a dense subset of $L^2(\mathbb{R})$), the above characterisation is more general. \emph{A fortiori}, for a POVM, the probability measures described by \eqref{eq:born} are typically not localized in any open set; there exist POVMs for which there is no state satisfying $p_{\omega}^{\E}(X)=1$ for any $X \neq \Sigma$. There are, however, POVMs which `almost' have the localizability property enjoyed by all PVMs:
		
		\begin{definition}[Norm-$1$ property]\label{def:norm1}
			A POVM $\E : \mathcal{B}(\Sigma) \to B(\hi)$ satisfies the \emph{norm-$1$ property} (see e.g. \cite{heinonen2003norm}) if for all $X \in \mathcal{B}(\Sigma)$ for which $\E(X) \neq 0$, it holds that $\no{\E (X)}=1$.  Such POVMs are called \emph{localizable}.
		\end{definition}
		
		\begin{proposition}\label{prop:norm1eq}
			The following are equivalent \cite{heinonen2003norm}:
			\begin{enumerate}
				\item $\E$ is  localizable 
				\item For every $X$ for which $\E(X)\neq 0$, there is a sequence of unit vectors $(\varphi_n)\subset \hi$
				for which $\lim_{n \to \infty}\ip{\varphi_n}{\E(X)\varphi_n}=1$. 
				\item For every $\E(X)\neq 0$ and for any $\epsilon > 0$ there exists a unit vector $\varphi_\epsilon \in \hi$ for which $\ip{\varphi_\epsilon}{\E(X)\varphi_\epsilon}>1-\epsilon$ (this is called the $\epsilon$-\emph{decidability property}).
			\end{enumerate}
		\end{proposition}
		
		A useful consequence is the following: 
		\begin{proposition}\label{prop:locseqgen}
			If a POVM $\E : \mathcal{B}(\Sigma) \to B(\hi)$ satisfies the norm-1 property and $\Sigma$ is metrizable, then for any $x \in \Sigma$ there exists a sequence of pure states $(\omega_n)$ such that the sequence of probability measures $(\mu^\E_{\omega_n})$ converges weakly to the Dirac measure $\delta_x $.
		\end{proposition}
		
		\begin{proof}
			Fix $x\in \Sigma $ and denote by $B_n $ the open ball centred at $x$ of radius $\frac{1}{n}$. Since $\E$ satisfies the norm-$1$ property, using the $\epsilon$-decidability property of Prop. \ref{prop:norm1eq} we can choose unit vectors $\varphi_n $ such that $\braket{\varphi_n }{\E(B_n) \varphi_n }> 1 - 1/n $. Denoting by $\omega_n $ the associated pure state $\omega_n = \dyad{\varphi_n }{\varphi_n }$, we have $\tr[\omega_n E(B_n )]= \braket{\varphi_n }{\E(B_n) \varphi_n }> 1 - 1/n $, and thus $ \mu^\E_{\omega_n}(B_n ) > 1 - 1/n $.
			
			We will show weak convergence using the portemanteau theorem \cite{billingsley2013convergence}. We must show that for each measurable set $X$ with negligible boundary (i.e., such that $\delta_x (\partial X) = 0 $ for all $x$) we have: $\lim\limits_{n\to \infty} \mu^\E_{\omega_n}(X) = \delta_x (X) $. We compute:
			
			\[\mu^\E_{\omega_n}(X) = \mu^\E_{\omega_n}(X\setminus B_n ) + \mu^\E_{\omega_n}(X\cap B_n ) .\]
			
			For the first term, we have $\mu^\E_{\omega_n}(X\setminus B_n ) \leq \mu^\E_{\omega_n}(\Sigma \setminus B_n ) = 1-\mu^\E_{\omega_n}(B_n ) \leq \frac{1}{n}$, which vanishes as $n$ goes to infinity. \\
			
			For the second term we distinguish two cases:
			
			\begin{itemize}
				\item If $ x\in X $, by hypothesis we may assume that $x\notin \partial X $ so $x\in \mathring{X} $, the interior of $X$. Since $\mathring{X}$ is open, for $n$ large enough we always have $B_n \subseteq \mathring{X} \subseteq X $, so $X\cap B_n = B_n $, and therefore $\mu^\E_{\omega_n}(X\cap B_n )= \mu^\E_{\omega_n}( B_n )> 1 - 1/n $. Thus, the second term goes to $1$ as $n$ goes to infinity.
				\item If $ x\notin X $, by hypothesis we may assume that $x\notin \partial X $ so $x\in \Sigma \setminus \overline{X} $, the complement of the adherence of $X$, which is an open set. Again for $n$ large enough we always have $B_n \subseteq \Sigma \setminus \overline{X} \subseteq \Sigma \setminus X $, hence $X\cap B_n = \emptyset $, leading to $ \mu^\E_{\omega_n}(X\cap B_n )= 0 $. Thus, the second term goes to $0$ as $n$ goes to infinity.
			\end{itemize}
			
			We have thus shown that $\lim\limits_{n\to \infty} \mu^\E_{\omega_n}(X) = \delta_x (X)$ for each measurable set $X$ with negligible boundary, and the portemanteau theorem \cite{billingsley2013convergence} demonstrates the claim.
		\end{proof}
		
		A converse to this theorem is given in \cite{JGthesis}.
		Thus if $\Sigma$ is metrizable and $\E$ is localizable, as will be the case in the sequel, we can approximate the Dirac delta measure centred at any $x \in \Sigma$ with measures of the form $\mu^\E_{\omega_n}(X) = \tr[\E(X)\omega_n(x)]$. We will call $\omega_n(x)$ a \emph{localizing sequence centered at $x$}.

		\subsection{Covariance}
		
		Observables are often characterised by their covariance properties (e.g. \cite{busch1997operational}):
		\begin{definition}\label{def:imp}
			Let $\E:\mathcal{B}(\Sigma) \to B(\hi)$ be a POVM, $G$ a locally compact second countable topological group, $\alpha : G \times \Sigma \to \Sigma$ a continuous transitive action (so that $\Sigma$ is a homogeneous $G$-space) and $U:G \to B(\mathcal{H})$ a strongly continuous projective unitary representation. Then $(U,\E,\hi)$ is called a \emph{system of covariance based on} $\Sigma$ if for all $X \in \mathcal{B}(\Sigma)$ and all~$g \in G$, 
			\begin{equation}\label{eq:covp}
				\E (\alpha (g, X))= U(g) \E (X) U(g)^*.
			\end{equation}
			$\E$ is called a \emph{covariant POVM}, and if $\E$ is projection-valued, then $(U,\E\equiv \mathsf{P},\hi)$ is called a \emph{system of imprimitivity (SOI)}. 
		\end{definition}
		We will often write $g.X$ to stand for $\alpha (g, X)$, and will presume that $U$ as given above is a true unitary representation.

		Systems of imprimitivity are characterised by the imprimitivity theorem, 
		which states that for a closed subgroup $H \subset G$ and $\Sigma = G/H$ with left $G$-action,  there is  (up to unitary equivalence) a one-to-one correspondence between systems of imprimitivity $(U,\mathsf{P},\hi)$ based on $\Sigma$ and unitary representations $U_{\chi}$ of $H$ (e.g. \cite{landsman2006between}). Irreducible systems of imprimitivity (those with no invariant subspaces of $(U,\mathsf{P})$) correspond to irreducible representations of $H$.\footnote{The exact correspondence is as follows. Given a representation 
			$U_{\chi}$ of $H$, construct the SOI $(U^{\chi}, \mathsf{P}^{\chi},\hi^{\chi})$, with 
			$\hi^{\chi}= L^2(G/H)\otimes \hi_{\chi}$ as the carrier space for
			the representation $U^{\chi}$ of $G$ induced by the representation $U_{\chi}$ of $H$, and $\mathsf{P}^{\chi}$ acts as multiplication by the characteristic function on $\hi^{\chi}$. In the other direction, for any SOI $(U, \mathsf{P},\hi)$, there is a unitary representation $U_{\chi}$ of $H$ for which $(U, \mathsf{P},\hi)$ is unitarily equivalent to the one given above. We note that any space $\Sigma$ with a continuous transitive $G$-action can be written as $G/H$, where $H=H_x$ is the stabiliser of some point $x \in \Sigma$.}
		
		For $G$ finite the canonical irreducible system of imprimitivity based on $G$ can be described very explicitly.
		
		\begin{example}
			\normalfont
			Let $L^2(G)$ denote the Hilbert space of the complex-valued functions $G \to \mathbb{C}$ on a finite group $G$ (sometimes denoted $\mathbb{C}[G]$), which has as an orthonormal basis the collection of indicator functions $\delta_g$.  
			These define the rank-$1$ projections $P(g)$, and to make contact with more common notation in the physics literature we write $\delta_g\equiv \ket{g}$ and  $P(g)\equiv \dyad{g}$. The left regular representation is given as $U_L(g)\ket{g'} = \ket{gg'}$. Then with $P: g \mapsto \dyad{g} \in B(L^2(G))$, $(U_L,P,L^2(G))$ is a system of imprimitivity based on $G$, and up to unitary equivalence is the unique irreducible one.
		\end{example}
		
		The simple case above generalises to the following.
		
		\begin{example}\label{ex:classical}
			\normalfont
			Let $G$ be a locally compact group acting transitively from the left on the measure space
			$(\Sigma, \mu)$, with $\mu$ a $\sigma$-finite invariant measure on $\Sigma$,
			and take
			$\mathcal{H}=L^2(\Sigma, \mu)$. Then $(U,P,\mathcal{H})$ as given below is a system of imprimitivity based on $\Sigma$.
			\begin{equation}\label{eq:csi}
				\begin{aligned}
					P(X)f = \chi_X f \\
					(U(g)f)(x)=f(g^{-1}.x).
				\end{aligned}
			\end{equation}
			
			\begin{remark}
				Setting $G=\Sigma$ as a topological space (so that $H=\{e\}$) yields the left-regular representation of $G$, equipped with the canonical (irreducible) system of imprimitivity based on $G$ as a left $G$-space. Further specialising to $G=\mathbb{R}$, with $\mu$ Lebesgue measure and $\mathbb{R}$ understood as the configuration space of a single particle, \eqref{eq:csi} yields the standard Schr\"{o}dinger representation of wave mechanics for $P=P^Q$ ($Q$ position) and $U$ is therefore generated by momentum. The uniqueness (all up to unitary equivalence) of the irreducible representation
				of $\{e\}$ corresponds to the uniqueness of the canonical commutation relation, and the imprimitivity point of view therefore constitutes a generalisation of the Stone-von Neumann theorem. The uniqueness of the shift-covariant spectral measure of position (which follows from the imprimitivity theorem)
				is particular to the sharp case; there exist many shift-covariant unsharp observables which can be obtained through convolution with a Markov kernel (e.g. \cite{Busch2016a}).\footnote{It is possible to be more general in the following sense. Let $(\mathcal{A},G,\alpha)$ be a $W^*$ dynamical system, that is, $\mathcal{A}$ is a $W^*$ algebra, $G$ a locally compact group and $\alpha:G \to Aut(\mathcal{A})$ an ultraweakly continuous homomorphism to the group of automorphisms of $\mathcal{A}$. A covariant representation of $\mathcal{A}$ is a pair $(\pi,U)$, where  $\pi : \mathcal{A} \to B(\hi)$ is
					a linear $*$-homomorphism $\mathcal{A} \to B(\hi)$ and $U$ is a unitary representation for which for all $g$ and $A$,
					\begin{equation}
						\pi(\alpha_g (A))=U(g)\pi(A)U(g)^*.
					\end{equation}
					If $\mathcal{A}=L^{\infty}(G,\mu)$, a representation of $\mathcal{A}$ corresponds exactly to a PVM, and a covariant representation to a system of imprimitivity. Example \ref{ex:classical} is recovered by choosing the representation $f \mapsto M_f$ defined by $(M_f) \varphi = f \varphi$. This suggests that the essential structure of the approach to quantum reference frames given here is displayed at the algebraic level.}
			\end{remark}
		\end{example}
		
		\begin{example}[Systems of coherent states]\label{ex:csspovm}
			\normalfont
			Systems of covariance can be constructed from certain families of generalized coherent states (see e.g. \cite{Perelomov,ali2000coherent} for coherent states, and \cite{de2021perspective} for an example of their use as quantum reference frames). For instance, set $U$ to be a unitary representation of a locally compact group $G$ in $\hi$, with a cyclic vector $\ket{\eta}$ (i.e., the span of $\{U(g) \eta\}$ is dense). Then the orbit $\{\eta_g := U(g) \eta\}$ is called a system of (Perelomov-Gilmore) coherent states, and if they satisfy a certain square integrability condition, they resolve the identity in the sense that
			\begin{equation}\label{eq:pgco}
				\int_G \dyad{\eta_g} d \mu (g) = \lambda \id,
			\end{equation}
			where $\lambda$ is some positive real number and $\mu$ is the normalized Haar measure. Then
			$$\E^{\eta}(X):= \frac{1}{\lambda}\int_X \dyad{\eta_g} d \mu (g)$$
			defines a covariant POVM, and therefore $(U,\E^{\eta},\hi)$ is a system of covariance. $\E^{\eta}$ does not satisfy the norm-$1$ property unless $\lambda = 1$, and rigorously speaking the existence of sharp coherent state systems, understood in the Hilbert space sense, requires $G$ to be discrete.
			
			We note that this definition can be made much more general, and we refer to  \cite{ali2000coherent} for a precise treatment. A minor generalisation is to proceed as above, but set $\Sigma=G/H$, with $H$ the stabilisier (up to a phase) of $\eta$, and define another system of coherent states $\eta_{\sigma(x)} = U(\sigma(x))\eta$, with $\sigma:\Sigma \to G$ any Borel section. Note that the measure $\mu$ in \eqref{eq:pgco} must be replaced by a quasi-invariant measure, and if this is actually invariant, $H$ is compact. This setting is explored in the perspective-neutral approach to quantum reference frames in \cite{de2021perspective} and has some interesting physical consequences, which are beyond the scope of the present work but will be addressed systematically within the framework presented in this paper in the future.
			
		\end{example}
		
		\begin{example}[Canonical Phase]
			\normalfont
			Let $\{\ket{n} \in \hi;n \in \mathbb{N}\}$ be an orthonormal basis of $\hi$, and  $N=\sum_{n \geq 0} n \dyad{n}$ a (densely defined) \emph{number observable}. Then 
			$\E : \mathcal{B}((0,2 \pi]) \to B(\hi)$ is a \emph{covariant phase observable} if it 
			satisfies \eqref{eq:covp} - explicitly, if $e^{iN\theta}\E(X)e^{-iN\theta} = \E(X+\theta)$, where addition is mod-$2\pi$. $U(\theta)=e^{iN\theta}$ is a strongly continuous unitary representation of the circle group, and 
			$(U,\E,\hi)$ forms a system of covariance. $\E$ is used to model the phase (observable) of an electromagnetic field, with the states $\{\ket{n}\}$ corresponding to photon number. These observables are completely characterised and take the form 
			\begin{equation}
				\E(X)= \sum_{n,m=1}^{\infty} c_{n,m}\int_{X}e^{i\theta(n-m)}\dyad{n}{m}\frac{d\theta}{2\pi},
			\end{equation}
			where $(c_{n,m})$ is a positive matrix with $c_{n,n}=1$ for all $n$. The boundedness from below of $N$ means that $\E$ is never sharp \cite{lahti1999covariant}, but the \emph{canonical phase} characterised by $c_{n,m}=1$ for all $n,m$ does satisfy the norm-1 property (\cite{lahti2000characterizations}).
		\end{example}
		
		We note that in general there are obstacles to the existence of covariant POVMs with good localization properties, for example as given in the following proposition:
		\begin{proposition}
			Let $U$ be a strongly continuous unitary representation of $G$ in $\hi$, where $G$ is is compact, abelian, and generated by $N = \sum n \dyad{n}$, and acts on itself from the left. Moreover let $\hi$ have finite dimension $d$. Then, for Haar measure $\mu$, any covariant $E$ and
			$X \in \mathcal{B}(G)$ ($X \neq G$), it holds that
			\begin{equation}
				\tr[\rho \E(X)] \leq d. \mu(X),
			\end{equation}
		\end{proposition}
		The proof is straightforward: it always holds that $\tr[\rho \, \E(X)] \leq \tr[\E(X)]$, and since $\ip{n}{\E(X)n} = \ip{n}{\E(g.X)n}$ and therefore 
		$\ip{n}{\E(X)n} = \mu(X)$ and $\tr[E(X)] = d. \mu(X)$ and the result follows. This means that for small $X$, a large $d$ is needed to have localization probability close to~$1$. 
		
		\section{Relational Quantum Kinematics}\label{sec:rqk}
		
		In this section, we introduce further ideas on which this operational approach to quantum reference frames is built. First, quantum reference frames are introduced, defined exactly as systems of covariance. As non-invariant objects, the POVMs defining the frame are \emph{not} observable themselves, but understood as part of the overall  description reflecting the experimental arrangement. This understanding suggests the notion of \emph{framing}, introduced next, which  builds in the choice of frame observable for the given experiment---the framed observables are those that can be realized as joint observables on the system-frame composite upon fixing the frame POVM. After this we recall the \emph{relativization maps} defined by a choice of frame observable, denoted by $\Y^\R$ (with $\R$ the frame), which are understood as yielding frame-relative observables, the collection of which is denoted by $B(\his)^\R \subseteq B(\hisr)^G$, given as the (ultraweak closure of the) image of the relativization map. These observables are both framed and invariant and are referred to as \emph{relative observables}. The relative state spaces are then taken to be the operational (quotient) state spaces associated with $B(\his)^\R$, denoted by $\S(\his)_\R := \S(\hisr)/\hspace{-3pt}\sim_{B(\his)^\R}$ (later further abbreviated to $\S(\hisr)/\hspace{-3pt}\sim_{\R}$). We also provide an equivalent characterization of these state spaces as the images of the predual maps $\Y^\R_*$ of $\Y^\R$, denoted by $\S(\his)_\R \cong \S(\his)^\R \subseteq \S(\his)$; this perspective is in keeping with the notions of relative states considered in other works, but has clear operational motivation and is mathematically precise.
		
		From here we analyze the structure of the spaces of relative states when the state of the frame is fixed; these are referred to as $\omega$-\emph{product relative states}, for $\omega \in \S(\hir)$ (a state of the frame). The relationships between the various spaces of relative observables and relative states are depicted in a dual pair of commuting diagrams, summarising the framework developed so far. Various properties of product relative states are then discussed. Finally, we turn our attention to the case of the state of the reference being highly localized, in order to show that the relational framework reproduces, at the probabilistic level, and in a limiting sense, the predictions of orthodox quantum theory: generalizing previous results, we prove the equivalence between relative and `absolute' kinematical descriptions in the context of arbitrary localizable principal (i.e., defined on the principal $G$-space) frames.
		
		\subsection{Quantum Reference Frames}\label{sec:qrfs}
		
		Given the presence of a symmetry corresponding to a locally compact group $G$ and unitary representation $U$, we recall that `truly' observable quantities must be invariant under $U$. Given that typically non-invariant quantities of some system $\Sy$ are used in the description of laboratory experiments, where one would at least envisage that e.g. Galilei symmetry should be respected, there should be an adequate explanation of the apparent contradiction thus arising. We propose that the answer is that those non-invariant quantities are notational shorthand for the true, invariant quantities, which actually describe system $\Sy$ and frame $\R$, but that typically the frame has been implicitly externalised and is not part of the theoretical description. Here, we include the frame explicitly, noting that the physical system serving as the frame should satisfy some requirements beyond simply being a quantum physical object. One such desideratum is that the frame admits a POVM which is covariant under the representation $U$ of $G$ on the frame. There are good physical reasons for such a demand; for instance the presence of a non-trivial observable paired with a state gives a probabilistic assignment of frame `orientations', and the covariance 
		mirrors what one would expect for classical coordinate changes, including e.g. active-passive duality. In other words, as anticipated, quantum reference frames will be described by the systems of covariance already introduced. Moreover, as we shall see, the choice of covariant frame observables allows, in a natural way, for `observable-relative-to-frame' to be cast as in invariant quantity. (See also \cite{de2021perspective} and references therein for further motivation). It is worth having in mind the relative position observable $Q_{\S} \otimes \mathbb{1}_{\R} - \mathbb{1}_{\S} \otimes Q_{\R}$; this is a shift-invariant (unbounded) self-adjoint operator, with $Q_{\R}$ having a covariant spectral measure corresponding to the frame, relative to which $Q_{\S}$ attains its operational meaning. The covariant frame observable appears as part of the overall invariant description, and does not represent a truly observable quantity in its own right, pointing to the relational flavour of the entire scheme. As we shall see, we may fix a state of the frame, recovering a typically non-invariant description of $\Sy$ alone; the correspondence between the two descriptions will be analysed.

		We are now in a position to give the definition of a quantum reference frame suited to the purposes of this paper.
		
		\begin{definition}
			A \emph{quantum reference frame} $\R$ is a system of covariance $\R=(U_\R,\E_\R,\hir)$ on a homogeneous (left) $G$-space $\Sigma_\R$.
		\end{definition}
		
		When referring to a frame $\R$ we assume $G$ is given and often refer to $\E_\R$ as the frame (observable) when $U_\R$ and $\hir$ are understood, and drop the $\R$ subscript in $\Sigma_\R$ which we often identify with a quotient space $G/H$. When $H$ is trivial, such an observable has a direct interpretation as an `orientation' observable. It is useful to be able to know when two quantum reference frames have no essential physical distinction, captured by the following notion of equivalence:
		
		\begin{definition}
			Two frames $\R_1 = (U_1,\E_1,\hit)$ and $\R_2=(U_2,\E_2,\hit)$ defined on the  value space $\Sigma$ are called \emph{(unitarily) equivalent} if there is a unitary map $U:\hit \to \hit$ such that for any $X \in \mathcal{B}(\Sigma)$ we have
			\[
			\E_2(X) = U\E_1(X)U^*.
			\]
			
		\end{definition}

		There are various special classes of frames:
		
		\begin{definition}
			~
			
			\begin{itemize} 
				\item A frame  $\R$ is called \emph{principal} if $\Sigma_\R$ is a principal homogeneous space, \emph{non-principal} otherwise.
				\item A frame  $\R$ is called \emph{sharp} if $\E_\R$ is sharp, unsharp otherwise.
				\item A frame  $\R$ is called \emph{ideal} if it is principal and sharp.
				\item A frame  $\R$ is called \emph{localizable} if $\E_\R$ is localizable.
				\item A frame  $\R$ is called \emph{complete} if there is no (non-trivial) subgroup $H_0 \subseteq G$ acting trivially on all the effects of $\E_\R$, and \emph{incomplete} otherwise. Such an $H_0$ will be called an \emph{isotropy subgroup} for $\E_{\R}$.
				\item A frame $\R$ is called a \emph{coherent state frame} if $\E_\R$ is constructed from a coherent states system as~in~\eqref{ex:csspovm}.
			\end{itemize}
			
		\end{definition}
		Some remarks are in order. We have chosen to define quantum reference frames on a homogeneous space rather than $G$ itself as is common elsewhere (see e.g., \cite{de2021perspective}, and \cite{mazzucchi2001observables} for a construction of covariant frame observables on the Poincar\'e group). The definition of completeness is reminiscent of that appearing in other works (e.g. \cite{Bartlett2007,de2021perspective}), and agrees with that of \cite{de2021perspective} in the case of coherent state frames. This can be contrasted with the non-principal setting, which is not considered in other works, but analysed for homogeneous spaces in \cite{glowacki2023quantum,fewster2024quantum}, in which there is a non-trivial isotropy group $H<G$ for the action of $G$ on $\Sigma$. General connections between incomplete frames and non-principal frames is the topic of current investigation.\footnote{We note here that for a localizable frame and any $h \in H_0$ in the isotropy subgroup, the value space of $\Sigma_{\R} = G/H$, and $(\omega_n)$ a sequence localizing at the identity coset $eH \in G/H$, we have
			\begin{align*}
				\delta_{hH}(X) &= \lim_{n \to \infty} \mu_{\omega_n(hH)}^{\E_\R}(X) = \lim_{n \to \infty}\tr[\omega_n.h^{-1} \, \E_\R] =
				\lim_{n \to \infty} \tr[\omega_n \, h^{-1}.\E_\R(X)]\\
				&= \lim_{n \to \infty} \tr[\omega_n \, \E_\R(X)]  = \delta_{eH}(X),
			\end{align*}
			so that $hH=eH$ for any $h \in H_0$, and we can conclude that $H_0 < H$. In particular, localizable principal frames are complete, and in general, a localizable frame observable on $\Sigma_\R = G/H$ factorizes through $B(\hir)^H$, i.e., we can write
			\[
			\E_\R : \mathcal{B}(G/H) \to B(\hir)^H \hookrightarrow B(\hir).
			\]
			}
		Sharp frames are described by systems of imprimitivity, which by the imprimitivity theorem, are unitarily equivalent to $L^2(G/H)$. Fixing $H=\{e\}$ and demanding that $(U,\mathsf{P},\hi)$ is irreducible yields (up to unitary equivalence) the left regular representation. The term `ideal' frames (and also `perfect'  frames \cite{Bartlett2007}) is used in the literature to mean something similar but not identical, the typical provision being that there is a collection of states in $\hi$ indexed by $g \in G$ (for example coherent state systems as in \cite{de2021perspective}), and these states must be orthogonal (or `perfectly distinguishable') for $g \neq g'$ for the frame to be ideal. However, we avoid this usage since the notion of orthogonality is not  compatible with the Hilbert space framework within which we are working - $G$ is generally continuous, yet $\hi$ is assumed to be separable.  We note also that elsewhere, coherent states themselves are understood, through their $G$-dependence, as encoding `frame orientations'; for us, this comes probabilistically through the measure on $G$ defined by an arbitrary state paired with $\E_{\R}$, through the Born rule \eqref{eq:born}, which we view as being more operationally motivated. There is no operational distinction between sharp frames and localizable frames based on the same space, and this is manifested clearly in the next subsection in which standard and relative observables are operationally compared. 
		There are various other interesting frames one can consider that we do not touch upon here, for instance \emph{informationally complete} frames, in which $\E$ is an informationally complete POVM, which is necessarily unsharp (e.g. \cite{riera2024uncertainty}).

		\subsection{Framing}

		After fixing a frame observable $\E_\R$, and for now \textit{not} imposing a symmetry constraint, one can consider \emph{joint observables} of system and frame (e.g. \cite{busch1997operational}). From here, one can compute conditional probabilities, i.e., the probability of observing the system observable $\E_\S:\mathcal{B}(\Sigma_\S) \to B(\his)$ to have value in some $Y \subseteq \Sigma_\S$
		given that the frame observable is measured to have value in some $X \subseteq \Sigma_\R$. This captures the idea that a frame is part of an experimental arrangement, and those observables of system-plus-frame which can be measured in that experiment are limited by the choice of frame. In the sequel, we will describe such objects in a relational way. The simplest joint observables are of the form $M(X,Y):=\E_\R(X) \otimes \E_\S(Y)$.
		In order to allow for mixing we take the closed convex hull of the 
		above joint observables, for a fixed frame but arbitrary system observable, whence:
		
		\begin{definition}
			Given a frame $\R$ and a system $\S$, the effects in
			\[
			\Eff(\hisr)^{\E_\R} := \rm{conv}\left\{\F_\S \otimes \E_\R(X) \hspace{3pt} | \hspace{3pt}\hspace{3pt} X \in \mathcal{B}(\Sigma_\R), \F_\S \in \Eff(\his)\right\}^{\rm cl},
			\]
			with $\rm{conv}$ denoting the convex hull, will be called $\E_\R$-\emph{framed effects}, while the elements of the Banach space 
			\[
			B(\hisr)^{\E_\R} := \rm{span}\left\{\F_\S \otimes \E_\R(X) \hspace{3pt} | \hspace{3pt}\hspace{3pt} X \in \mathcal{B}(\Sigma_\R), \F_\S \in \Eff(\his)\right\}^{\rm cl}
			\]
			will be referred to as $\E_\R$-\emph{framed observables}. We denote by $\sim_{\E_\R}$ the operational equivalence relation on $\T(\hisr)$ taken with respect to $\Eff(\hisr)^{\E_\R}$ (or, equivalently, with respect to $B(\hisr)^{\E_\R}$). The space of $\E_\R$-\emph{framed trace class operators} is given by
			\[
			\T(\his)_{\E_\R} := \T(\hisr)/\hspace{-3pt}\sim_{\E_\R},
			\]
			while the operational state space of $\E_\R$-\emph{framed states} is defined to be
			\[
			\S(\his)_{\E_\R} := \S(\hisr)/\hspace{-3pt}\sim_{\E_\R}.
			\]
		\end{definition}
		
		The $\E_\R$-framed (or just `framed' when the frame is clear from the context) states are then precisely the equivalence classes of density operators which cannot be distinguished by framed observables.
		
		Prop. \ref{prop:statespace} and \ref{prop:iml} give the following.
		\begin{proposition}
			There is a Banach space isomorphism
			\[
			B(\his)^{\E_\R}_* \cong \T(\his)_{\E_\R}.
			\]
			Moreover, the subset $\S(\his)_{\E_\R} \subset \T(\his)^{\rm sa}_{\E_\R}$ is a total convex state space.
		\end{proposition}
		
		We finish by noting that, upon the choice of the frame observable and imposing framing, the states $[\omega]_{\E_\R} \in \S(\hir)/\hspace{-3pt}\sim_{\E_\R}$ are in bijection with probability measures $\mu^{\E_\R}_\omega$ they give rise to upon evaluation of the frame observable. In this sense the frames, when considered in isolation, can be described in rather classical terms; it is this sort of `classicality' that leads to a frame change that differs from others appearing in the literature.
		
		\subsection{Relative Observables and States}\label{sec:relobs}
		
		Quantum reference frames are introduced, in analogy to classical physics, to study properties of some system relative to the given frame. Together they are described as a compound system.
		The main principle put forward in \cite{Loveridge2017a,Loveridge2020a,Miyadera2015e,loveridge2017relativity} is that truly observable quantities are invariant under gauge/symmetry transformations, understood here as a diagonal action on the compound system. We note that other choices are possible; as we have mentioned, the quantum information approach \cite{Bartlett2007} puts the invariance on the states, which is problematic for non-compact groups, and the perspective-neutral approach demands that the relevant Hilbert space is formed from invariant vectors. We reflect further on this when we compare our frame change maps to those appearing in other approaches in  Sec. \ref{sec:fcp}.
		
		In this work, we make a further distinction between the invariant algebra $B(\hisr)^G$, which contains observables/effects which, as we will see, may be understood without reference to a frame, and the observables in the image of the relativization map (defined below), which, as well as being invariant, are also framed, and therefore reflect some internal structure of the joint systems, and the choice of frame observable. They will be called \emph{relative observables/operators}.
		
		\subsubsection{Relativization and relative observables}

		Let $\R=(U_\R, \E_\R, \hir)$ be a principal frame and $U_\S$ be a strongly continuous, unitary representation of $G$ in $\his$. As has been shown in \cite{Loveridge2017a}, there is a map $\Y^{\R}: B(\his)\to B(\hisr)$ defined by
		\begin{equation}
			\Y^\R(A_ \S) := \int_G U_\S(g)A_\S U_\S(g)^* \otimes d\E_\R(g),
		\end{equation}
		called the relativization map, which has the following properties (not all of which are independent): $\Y^\R$ is linear, unital, adjoint-preserving, preserves effects, bounded (thus continuous), completely positive, a contraction, normal, injective (and hence isometric) if $\E_\R$ satisfies the norm-1 property (see  \ref{prop:yeninjective} below), and multiplicative exactly when $\E_\R$ is projection valued (sharp). In the setting that  $\E_\R$ is sharp, the image $\rm{Im}(\Y^\R)$ is a von Neumann algebra isomorphic to $B(\his)$, making $\Y^\R$ a faithful representation of $B(\his)$ in $B(\hisr)$. Crucially, the operators in the image of $\Y^\R$ are framed (see Prop. 3.7.1. in \cite{JGthesis} for the proof), and invariant $\Y^\R(B(\his)) \subset B(\hisr)^{G}$ (which can be verified by direct computation).
		
		The $\Y^{\R}$ construction \cite{Loveridge2017a} was developed as a generalization, and making rigorous, of the $\$$ map \cite{Bartlett2007}, and was used to construct particular invariants;
		the physical difference between observables in the image of $\Y^\R$ and general invariants was not discussed at that time but will play an important role here. The twirl map arises as a special case of this construction when $\hir$ is taken to be $\mathbb{C}$, with (necessarily) trivial $G$ action. Indeed, the notion of a covariant POVM then coincides with that of a normalized invariant measure (Haar measure), and thus there is exactly one when $G$ is compact, and none otherwise. Another simple example is when $\R$ is taken to be the canonical irreducible system of imprimitivity for a finite group $G$, in which case $\Y^{\R}$ reads
		\[
		\Y^{\R}(A_\S) = \sum_{g \in G}  U_{\S}(g)A_\S U_{\S}(g)^* \otimes \dyad{g}_{\R}.
		\]
		The relativization in some sense replaces the naive `invariantization' given by the twirl and applies to the general case of locally compact groups and arbitrary (strongly continuous, unitary) representations. The definition of a relational Dirac observable as in \cite{de2021perspective} is recovered through the twirl map by $ \Y^{\R}(g.A_S)$ (for $g \in G$) when $\E_\R$ is a covariant POVM associated to a coherent state system. Thus in the case of coherent state frames, the set of relational Dirac observables and relative self-adjoint operators are the same.
		
		Physically, $\Y^\R$ is understood as providing an explicit inclusion of the frame $\R = (U_\R, \E_\R, \hir)$ into the description of $\S$; we view $\Y^\R(A_\S)$ as a \emph{relativization} of $A_\S$. This naturally extends to POVMs by composition. For instance, for $G=\mathbb{R}$, and fixing the appropriate frame observable to be the spectral measure of the position operator, $\Y^\R$ takes the position observable $Q_\S $ to $Q_\S  \otimes \id_\R - \id_\S \otimes Q_R$ (which can be proven at the level of the respective spectral measures); it also produces unsharp relative time observables, relative angle and phase observables for $G=U(1)$  (\cite{loveridge2012quantum,Loveridge2017a,loveridge2019relative}). We take the view that what is measured is the relation between system and frame, or more precisely, the relative observable, such as the relative position. This is in line with, but stronger than, the requirement that observable quantities are (gauge-)invariant. The covariant POVM of the frame is never directly measured (as a non-invariant object), but as we shall see, can be justified as being a crucial ingredient in the description.
		
		\begin{definition}\label{def:Yrelobs}
			An operator $A \in B(\hisr)^G$ is called $\R$-\emph{relative} (or just 'relative' when the frame is clear from the context) if it belongs to the ultraweak closure of the image of $\Y^\R$. We write $$B(\his)^\R := \Y^\R(B(\his))^{\rm cl} \subset B(\hisr)^G$$ for the set of $\R$-relative operators. As usual, the term `relative observable' may be used for self-adjoint relative operators, POVMs of the form $\Y^\R \circ \E$, or as a collection the entire image of $\Y^\R$, as the situation dictates.
		\end{definition}
		
		By definition, $B(\his)^\R$ is ultraweakly closed, and hence also norm closed and therefore an \emph{operator space} (e.g. \cite{arveson2008operator}).  We note that for a sharp frame, $\Y^\R$ is a normal homomorphism and therefore the 
		image is automatically closed. Relative operators are clearly invariant but do not exhaust the invariant operators in general. In \cite{fewster2024quantum}, an \emph{extended relativization map} is constructed, which does yield all invariants, but is not uniquely defined by $\Y^{\R}$. See also \cite{glowacki2024relativization} for a recent  categorification of the $\Y^{\R}$ map and \cite{riera2024uncertainty} for the case of a projective representation of $\mathbb{R}^2$ acting as phase space translations.
		
		For any $A_\S \in B(\his)^G$, and any principal frame $\R$ we have $\Y^\R(A_\S) = A_\S \otimes \mathbb{1}_\R$, and thus $\Y^\R(B(\his)^G) \cong B(\his)^G$, whatever the frame. The invariant algebra can thus be understood as a frame-independent description -- this perspective will be strengthened when we consider states. Interesting examples of relative observables arise in the context of two frames.

		\begin{definition}\label{def:relorobs}
			For a pair of principal frames $\R_1$ and $\R_2$ define the \emph{observable of relative orientation} (of $\R_2$ with respect to $\R_1$), denoted by $\E_2*\E_1$, via the relativization map\footnote{
				The notation is chosen to reflect the observation that when $\his$ and $\hir$ are both taken to be $\mathbb{C}$, the formula reproduces standard convolution of measures on $G$.}
			\begin{equation}
				\E_2 * \E_1 :=\Y^{\R_1} \circ \E_2 = \int_G g.\E_2(\cdot) \otimes d\E_1(g).
			\end{equation}
		\end{definition}
		
		Indeed, many important examples of relative observables---position, time, phase and angle (\cite{loveridge2012quantum, Loveridge2017a})---are relative orientation observables. The name can be further justified by the following.
		\begin{proposition}
			For a pair of principal localizable frames $\R_1$ and $\R_2$ and corresponding localizing sequences $(\omega_n)$ and $(\rho_m)$ (centred at $e \in G$), writing $\Omega_{n,m}(h) := \rho_m.h^{-1} \otimes \omega_n \in \S(\hit \otimes \hio)$ with $h \in G$ we have
			\[
			\lim_{n,m\to\infty} \mu_{\Omega_{n,m}(h)}^{\E_2 * \E_1} = \delta_h.
			\]
		\end{proposition}
		
		\begin{proof}
			We calculate
			\begin{align*}
				&\lim_{n,m\to\infty} \tr[(\rho_m.h^{-1} \otimes \omega_n) \int_G g.\E_2(X) \otimes d\E_1(g)]
				= \lim_{m\to\infty} \tr[\rho_m.h^{-1} \E_2(X)]\\
				&= \lim_{m\to\infty} \tr[\rho_m \E_2(h^{-1}.X)]
				= \delta_e(h^{-1}.X) = \delta_h(X),
			\end{align*}
			where we have used Prop. \ref{prop:locseqgen} twice.
		\end{proof}
		
		Thus when the first frame is localized at the identity, and the second at $h \in G$, the relative orientation observable gives a probability distribution localized at $h$, which is understood as their relative orientation.\footnote{We have $\lim_{m \to \infty}\mu^{\E_2}_{\rho_m.h^{-1}.} = \delta_h$, so that the states $\rho_m.h^{-1}$ approximate localization at $h$ for large $m$. From here, we see that $\lim_{n \to \infty} \mu_{\Omega_{n,m}(h)}^{\id_2 \otimes \E_1} = \delta_e$ and $\lim_{m \to \infty} \mu_{\Omega_{n,m}(h)}^{\E_2 \otimes \id_1} = \delta_h$.} Note that since $\E_2 * \E_1$ is~invariant, we could just as well evaluate it on e.g. $(\rho_m \otimes \omega_n.h)$, approximating the first frame localized at $h^{-1} \in G$, and the second at the identity, with the same result, and this the relative orientation is preserved under an overall (gauge) transformation. Notice also that if the roles of the frames are swapped, the probability distribution of the relative orientation observable evaluated on $\omega_n \otimes \rho_m.h^{-1}$ is localized at $h^{-1}$ as one would expect. Indeed, a simple calculation gives
		\[
		\E_2 * \E_1 (X) = {\rm SWAP} \circ \E_1 * \E_2 (X^{-1}),
		\]
		where ${\rm SWAP}$ switches the tensor product factors as in \cite{de2020quantum}, i.e., for $A_1 \otimes A_2 \in B(\hio \otimes \hit)$ we have ${\rm SWAP}(A_1 \otimes A_2) = A_2 \otimes A_1 \in B(\hit \otimes \hio)$, and $X^{-1}:= \{g \in G \hspace{2pt}| \hspace{2pt} g^{-1} \in X\}$.
		
		We view the fixing of the frame observable in direct analogy to choosing a pointer observable in the quantum theory of measurement (e.g. \cite{busch1997operational}), in which a particular observable of the apparatus is fixed, and by varying the initial pointer state, various different system observables may be measured, which may be equal to that which was intended, or may only be some approximation of it (the connection between quantum reference frames and quantum measurements under symmetry is explored in \cite{Loveridge2020a}). From a historical perspective, we view the fixing of a frame as being in line with the epistemology of Bohr \cite{bohr1996discussion}, who stressed that quantum phenomena must be understood in terms of the `entire experimental arrangement'. This gives some classical-like `window' into the quantum world, or more precisely, gives one manifestation of the relational world described by system-frame composite; different frames afford different possibilities (c.f.~\cite{landsman2017foundations}). 
		
		In more prosaic terms, it is only the relative observables that admit an interpretation of `system relative to  frame', which are therefore special amongst the invariant observables. For instance, the relative position or angle take the form of difference operators, which allow for clear separation of system and frame, and more generally the relative observables are exactly those for which system and frame have clear meaning. By contrast, the larger, invariant algebra (which contains operators that are not relative/difference operators), is exactly the collection of operators that may be defined without an external frame, and does not refer to `internal structure'.

		\subsubsection{Relative state spaces}
		
		Despite playing an essential role, the notion of `state relative to a frame' is left implicit/taken as primary in \cite{giacomini2019quantum}, where the `modern' notion of quantum reference frame transformation was first introduced. One objective of that work was, as we understand it, to provide frame changes without recourse to any external or classical reference. A definition of relative state is provided in \cite{de2020quantum} and the frame changes there are in line with \cite{giacomini2019quantum} (though applied also in more general settings); however, no operational justification is provided, and strictly speaking, the procedure applies only to countable $G$, since the frame state is assumed to be perfectly localized at $e \in G$. We now seek to rectify these various shortcomings, by providing an operationally motivated notion of relative state, which rigorously extends that of \cite{de2020quantum} to continuous (locally compact) groups, and provides an alternative conceptual foundation for the notion of relative state as it arises in the theory of QRFs; in particular we show that the localization at/near $e \in G$ is required to view states as relative states.
		
		\begin{definition}\label{def:Yrel}
			Given a frame $\R$ and a system $\S$, we will denote by $\sim_{\R}$ the operational equivalence relation on $\T(\hisr)$ taken with respect to the set of relative operators $B(\his)^\R$. The space of $\R$-relative (or just 'relative' when the frame is clear from the context) trace class operators is given by
			\[
			\T(\his)_\R := \T(\hisr)/\hspace{-3pt}\sim_{\R};
			\]
			the operational state space of $\R$-\emph{relative states} is defined to be
			\[
			\S(\his)_{\R} := \S(\hisr)/\hspace{-3pt}\sim_{\R}.
			\]
		\end{definition}
		
		In analogy to the $G$-equivalence relation, we see that the standard duality 
		$\T(\hi)^*\cong B(\hi)$ appears when these objects are replaced by their relative counterparts.
		
		\begin{proposition}\label{def:Yrelst}
			The Banach space $B(\his)^{\R}$ has the uniquely given predual 
			\[
			B(\his)^{\R}_* \cong \T(\his)_\R.
			\]
		\end{proposition}
		\begin{proof}
			Use Prop. \ref{prop:iml}. 
		\end{proof}
		
		Denoting the image of the predual map $\Y^\R_*$ restricted to states by
		\[
		\S(\his)^\R:=\Y^\R_*(\S(\hisr)) \subseteq \S(\his),
		\]
		we have the following.
		
		\begin{proposition}\label{prop:y}
			The set of $\R$-relative states is a total convex state space in $\T(\hisr)^{\rm{sa}}/\hspace{-3pt}\sim_{\R}$ and the map
			\[
			y_\R: S(\his)^\R \ni \Omega^\R \equiv \Y^\R_*(\Omega) \mapsto [\Omega]_\R \in \S(\his)_\R
			\]
			is well-defined and establishes the state space isomorphism $S(\his)^\R \cong S(\his)_\R$.
		\end{proposition}

		\begin{proof}
			Use Prop. \ref{prop:statespace} and \ref{generalst} for $\Lambda=\Y^\R$.
		\end{proof}
		
		Since this characterization of relative states will be used frequently,
		it is worth being explicit about how we use the above notation. Given an arbitrary $\Omega \in \S(\hisr)$, we write $\Omega^\R = \Y^\R_*(\Omega) \in \S(\his)$, and $[\Omega]_\R$ for the corresponding equivalence class on $\S(\hisr)$, i.e., the element of $\S(\hisr)/{\sim _{\R}}$.
		
		As a direct consequence of the invariance of the image of $\Y^\R$, the predual $\Y^\R_*$ map factorizes through $\T(\hisr)/\hspace{-3pt}\sim_G$ and thus
		\[
		\S(\his)^\R = \Y^\R_*(\S(\hisr)/\hspace{-3pt}\sim_G).
		\]
		For $G$ compact we then have $\Y^\R_*(\Omega)=\Y^\R_*(\mathcal{G}_* (\Omega))$; $\mathcal{G}_* (\omega \otimes \rho )$ is sometimes referred to as the 'relational encoding' of $\rho$ in the older quantum reference frames literature (c.f. \cite{palmer2014changing}). If the action of $G$ is understood as a gauge transformation, applying $\Y_*^\R$ to the state on $B(\hisr)^G$ can thus be viewed as providing a relative description whilst maintaining (gauge) invariance. We finally remark that the $\R$-equivalence relation is strictly stronger than the $G$-equivalence, leading to another distinct set of physical states that may have interesting properties.
		
		\subsection{Restriction and Conditioning}
		
		We now recall the restriction map (e.g. \cite{loveridge2017relativity}), which allows for the conditioning of observables of the system-plus-reference with a specified state of the reference.
		
		\begin{definition}\label{def:rest}
			Let $\omega \in \S(\hir)$ be any (normal) state of the reference. The \emph{restriction map}\footnote{Versions of this map have appeared elsewhere in the literature, often in a rather ill-defined form. Usually the implicit desired map is $B(\hisr) \ni A \mapsto \ip{\cdot \otimes \varphi_R}{A \cdot \otimes \varphi_R}: \his \times \his \to \mathbb{R}$ for $A=A^*$, which as a bounded real quadratic form defines uniquely a bounded self-adjoint $A_\S \in B(\his)$ which for all $\varphi_\S$ satisfies $\ip{\varphi_\S \otimes \varphi_\R  }{A \varphi_\S \otimes \varphi_\R }= \ip{\varphi_\S }{A_\S \varphi_\S} =: \ip{\varphi_\S}{\Phi_{\varphi_\R}(A) \varphi_\S }$. The generalisation of $\Phi_{\varphi_\R}$ to mixtures is then given by $\Gamma_{\omega}$ defined in \eqref{eq:res}.} $\Gamma_{\omega}: B(\hisr)\to B(\his)$ is given by
			\begin{equation}\label{eq:res}
				\tr[\rho \, \Gamma_{\omega}(A)] = \tr[(\rho \otimes \omega) \, A],
			\end{equation}
			holding for all $A \in B(\hisr)$ and $\rho \in \mathcal{S}(\his)$. Equivalently, it is the continuous linear extension to $B(\hisr)$ of
			\[
			\Gamma_{\omega}: A_\S \otimes A_\R \mapsto \tr[\omega \, A_\R] A_\S,
			\]
			or again equivalently the dual of the isometric embedding $\mathcal{V}_{\omega} : \rho  \mapsto \rho \otimes \omega $, 
			referred to as the $\omega$-\emph{product map}. 
		\end{definition}
		
		The restriction map $\Gamma_{\omega}$ is normal, completely positive and trace-preserving (hence trace norm continuous), and is a noncommutative conditional expectation (e.g., \cite{takesaki1972conditional}; see \cite{kadison2004non} for a comprehensive review). It is equivariant (covariant) exactly when $\omega$ is invariant. It is understood as providing a description of the system contingent on a particular state of the frame.
		
		\subsubsection{Conditioned relative observables}
		
		We now explore how fixing a state of the frame affects the description of $\Sy$ in terms of relative operators.
		
		\begin{definition}
			The map
			\[
			\Y^\R_\omega:= \Gamma_\omega \circ \Y^\R: B(\his) \to B(\his),
			\]
			will be called the $\omega$-\emph{conditioned $\R$-relativization map}.
		\end{definition}
		
		This map gives a (typically non-invariant) description inside $B(\his)$ of the (invariant) relative description in terms of $B(\his)^{\R}$, contingent on the state $\omega$ of the reference.
		As a composition of such maps, $\Y^\R_\omega$ is unital, normal and completely positive. The $\omega$-conditioned relative operators take the form:
		\begin{equation}\label{eq:croi}
			\Y^\R_\omega(A_\S) = \int_G U_\S(g) A_\S U_\S(g)^* \, d\mu^{\E_\R}_\omega(g),
		\end{equation}
		representing a `weighted average' of the operator $A_\S$ with respect to the probability distribution given through the evaluation of $\omega$ on the frame observable.
		Note that different $\omega$ may give rise to the same measure $\mu^{\E_\R}_\omega$, which in turn gives rise to the same $\omega$-conditioned relativization. Notice also that if $\omega$ is invariant, 
		$\Y^\R_\omega(A_\S) = \mathcal{G}(A)$, so that the $\omega$-conditioned relative operators are also invariant, leading to a symmetry constraint on system observables.
		
		\begin{definition}
			The operator space given by the ultraweak closure of ${\rm Im}\Y^\R_\omega$, i.e.,
			\[
			B(\his)^\R_\omega := \Y^\R_\omega(B(\his))^{\rm cl} \subseteq B(\his),
			\]
			is called the space of $\omega$-\emph{conditioned} $\R$-\emph{relative operators}.
		\end{definition}
		
		\subsubsection{Product relative states}
		
		Conditioned relative operators give rise to the corresponding operational state spaces, with the state-observable duality resulting in the notion of product relative states.
		
		\begin{definition}
			We will denote by $\sim_{(\R,\omega)}$ the operational equivalence relation on $\T(\his)$ taken with respect to  $B(\his)^\R_{\omega}$. The space of $\omega$-\emph{product} $\R$-\emph{relative trace class operators} is given by
			\[
			\T(\his)_\R^{\omega}:= \T(\his)/\hspace{-3pt}\sim_{(\R,\omega)},
			\]
			while the operational state space of $\omega$-\emph{product} $\R$-\emph{relative states} is defined to be
			\[
			\S(\his)_\R^{\omega} := \S(\his)/\hspace{-3pt}\sim_{(\R,\omega)}.
			\]
		\end{definition}
		For brevity we will refer to `product relative states' when $\R$ and $\omega$ are understood from context.

		\begin{proposition}
			For any (normal) state $\omega \in \S(\hir)$ the $\omega$-conditioned relative observables form a Banach space with the unique Banach predual given by
			\begin{equation}\label{eq:adrr}
				[B(\his)^\R_\omega]_* \cong \T(\his)_\R^{\omega}.
			\end{equation}
			Moreover, the set of $\omega$-product relative states $S(\his)^{\omega}_{\R}$ is a total convex state space in $\T(\his)_\R^{\omega}$, and there is a state space isomorphism
			\[
			S(\his)^{\omega}_{\R} \cong \mathcal{P}^\R_\omega (\S(\his)) \subseteq \S(\his).
			\]
		\end{proposition}
		
		\begin{proof}
			Use Prop. \ref{prop:iml}, \ref{prop:statespace} and \ref{generalst} with $\Lambda = \Y^\R_\omega$.
		\end{proof}
		
		The $\omega$-conditioned relative observables are understood as follows. 
		If we view the reference $\R$ as initially external to the system, meaning that it is not explicitly realised in the theoretical description but nevertheless plays some role in the observed phenomena/probability distributions, the relativization map $\Y^\R$ describes how some  $A_\S \in B(\his)$ is realized as an invariant  $\Y^\R(A_\S) \in B(\hisr)^G$ (naturally extending to POVMs) with respect to the frame $\R$. The restriction map then provides a description of $\S$ contingent on the given state preparation of $\R$, after which $\R$ can be viewed as external, arriving at the corresponding subset of observables of $\S$. We note (and discuss further later) that not all of $B(\his)$ is typically in the image of the $\omega$-conditioned relativization, and therefore the definition is not vacuous. Eq. \eqref{eq:adrr} captures again a duality between states and observables, now in the conditioned relative setting.
		
		Before we describe the properties of product relative states, we
		summarize graphically the maps and spaces thus far encountered as a dual pair of commutative diagrams. The "$i$" maps denote the inclusions of (ultraweakly closed) operator spaces, while "$\pi$" maps denote the corresponding quotient projections between the relevant predual spaces. We typically use `hooked' arrows for injective maps, and double-headed arrows for surjective maps.
		
		\vspace{10pt}
		\begin{center}
			\begin{tikzcd}
				&    &   & B(\hisr) &   \\
				&   & B(\his)^\R \arrow[rd, "\Gamma_\omega \circ i^\R"] \arrow[rr, "i^G_\R", hook] \arrow[ru, "i^\R", hook] &    & B(\hisr)^G \arrow[lu, "i^G"', hook'] \arrow[ld, "\Gamma_\omega \circ i^G"] \\
				& B(\his) \arrow[ru, "\Y^\R"] \arrow[rr, "\Y^\R_\omega", two heads] &  & B(\his)^\R_\omega &  
			\end{tikzcd}
		\end{center}
		\vspace{10pt}
		
		The predual of the above diagram, restricted to state spaces for the convenient representation of $\S(\his)_\R^\omega \subseteq \S(\his)$, reads
		\vspace{10pt}
		\begin{center}
			\begin{tikzcd}
				&     &   & \S(\hisr) \arrow[ld, "\pi_\R"', two heads] \arrow[rd, "\pi_G", two heads] &   \\
				&  & \S(\his)_\R \arrow[ld, "\Y^\R_*"', hook'] &    & \S(\hisr)_G \arrow[ll, "\pi_G^\R"', two heads] \\
				& \S(\his) &  & \S(\his)_\R^\omega \arrow[lu, "\pi_\R \circ \mathcal{V}_\omega"'] \arrow[ll, "(\Y^\R_\omega)_*"', hook'] \arrow[ru, "\pi_G \circ \mathcal{V}_\omega"]&     
			\end{tikzcd}
		\end{center}
		
		\subsubsection{Properties of product relative states}
		The $\omega$-product relative states take the form
		\begin{equation}
			\rho ^{(\omega)} := \Y^\R_*(\rho \otimes \omega) = \int_G U_\S(g)^*\rho U_\S(g) \, d\mu _{\omega}^{\E_\R}(g),
		\end{equation}
		where the measure $\mu_{\omega}^{\E_\R}$ is given as usual by $\mu_{\omega}^{\E_\R}(X):= \tr[\E_\R(X) \, \omega]$, and we have introduced the notation $\rho^{(\omega)}$ to indicate the particular frame state ($\omega$) that is used for conditioning. The product relative states arise from product states on the composite system; in the case of a localizable frame in a localized state they generalize the alignable states of \cite{krumm2021quantum} (we recall that $\Y^{\R}_*$ factors through the $G$-action). Notice again, as in Eq. \eqref{eq:croi}, that the state $\rho ^{(\omega)}$ depends only on the measure $\mu _{\omega}^{\E_\R}$ once $U_{\S}$ and $\rho$ are given.
		
		\begin{proposition}\label{prop:sco}
			Product relative states satisfy the following symmetry condition
			\begin{equation}
				\rho^{(\omega.h)} = (\rho.h^{-1})^{(\omega)}
			\end{equation}
		\end{proposition}
		
		\begin{proof}
			We calculate:
			\begin{align*}
				\rho^{(\omega.h)} &= \int_G \rho.g \, d\mu^{\E_\R}_{h.\omega}(g)= \int_G \rho.g \, d\mu^{\E_\R}_{\omega}(hg)
				= \int_G \rho.(h^{-1}g') \, d\mu^{\E_\R}_{\omega}(g')\\
				&= \int_G (\rho.h^{-1}).g' \, d\mu^{\E_\R}_{\omega}(g') = (\rho.h^{-1})^{(\omega)}.
			\end{align*}
			We have changed the integration variable $g'=hg$ and used the fact that $\mu^{\E_\R}_{\omega.h}(g) = \mu^{\E_\R}_{\omega}(hg)$ which follows directly from the covariance of $\E_\R$.
		\end{proof}
		
		Thus reorienting the frame by $h \in G$ is equivalent
		to reorienting the system by $h^{-1} \in G$. This represents a type of `active-versus-passive' equivalence for quantum frame reorientation.
		
		The next proposition demonstrates the plausible claim that invariant system states are defined without reference to an external frame or, more precisely, are independent from the chosen reference. 
		\begin{proposition}\label{prop:foai}
			Let $\rho \in \S(\his)^G$. Then $\rho ^{(\omega)} = \rho$ for any $\omega$, and any choice of frame $\R$.
		\end{proposition}
		
		\begin{proof}
			We calculate:
			\[
			\rho ^{(\omega)} = \int_G U_\S(g)^*\rho  U_\S(g) \, d\mu _{\omega}^{\E_\R}(g) = \int_G \rho \, d\mu _{\omega}^{\E_\R}(g) = \rho,
			\]
			where we only needed to use invariance of $\rho$ and normalization of $\E_\R$.
		\end{proof}
		We recall that the invariance of $\rho$ typically requires $G$ compact. It follows directly from Prop. \ref{prop:sco} that $\rho ^{(\omega.h)} = \rho^{(\omega)}$ if $\rho$ is invariant, showing directly how the relative description of $\S$ in terms of product relative states is not sensitive to frame reorientations if $\rho$ is invariant, as one would expect (invariant objects `look the same from all vantage points'). We note that the frame-independence of invariant states follows directly from our framework, and does not need to be stated as an assumption. Observe that Prop. \ref{prop:foai} can be immediately generalized in order to pertain to  non-compact $G$, by considering the class of $G$-indistinguishable density operators.
		\begin{proposition}
			Let $\rho \in \S(\his)$. Then $[\rho^{(\omega)}]_G=[\rho]_G$ for any $\omega$, and any choice of frame $\R$.
		\end{proposition}
		\begin{proof}
			Set $A\in B(\his)^G$. Then 
			\begin{align*}
				\tr[\rho^{(\omega)}A] &=\tr [ \int_G U_\S(g)^*\rho U_\S(g) A \, d\mu _{\omega}^{\E_\R}(g)]\\
				&= \int_G \tr [\rho \, U_\S(g) A U_\S(g)^*] \, d\mu _{\omega}^{\E_\R}(g) = \tr [\rho \, A], 
			\end{align*} 
			using the cyclicity of the trace, the invariance of $A$ and again the normalisation of the measure. 
		\end{proof}

		Another plausible intuition---that a reference in an invariant state can only give rise to invariant relative states---is confirmed by the following proposition.
		
		\begin{proposition}\label{Prop:soai}
			Suppose $\omega \in \S(\hir)^G$. Then $\rho ^{(\omega)} = \mathcal{G}(\rho )$ for any $\rho \in \S(\his)$.
		\end{proposition}
		
		\begin{proof}
			If $\omega$ is invariant, then $\mu^{\E_\R}_{\omega}$ is Haar measure, denoted $d\mu$ as before. We then have
			\begin{equation*}
				\rho ^{(\omega)} = \int_G U_\S(g)^*\rho U_\S(g) \, d\mu(g)= \mathcal{G}(\rho). 
			\end{equation*} 
		\end{proof}
		Thus if the reference is in an invariant state, the only relative states defined with respect to it are also invariant.
		
		Now consider $\hir = \mathbb{C}$. As already noted, any unitary representation
		on $\mathbb{C}$ is trivial, the only state is trivially invariant, and a covariant POVM $\E: \mathcal{B}(G) \to B(\mathbb{C}) \simeq \mathbb{C}$ is a measure $\mu_\E$ on $G$ satisfying $\mu_\E(g.X) = g.\mu_\E(X) = \mu_\E(X)$, and is thus the normalized Haar measure. There is a unique such covariant POVM iff $G$ is compact, and in that case, $\Y^\R_* = \mathcal{G}$.  Therefore, there will only be invariant states in the image. From this we conclude that in general not all states are relative states, and cannot be approximated by relative states.

		\subsection{Localizing the Reference}

		The structure of the $\omega$-conditioned $\R$-relative observables is dictated by the frame $\R$ and the state $\omega$, as has been investigated in
		e.g. \cite{Loveridge2017a,Miyadera2015e}. Those works analyzed the agreement between the standard description
		of $\S$ and the relative one given on the compound system $\S$ and $\R$, in which it was found that the localization/delocalization of $\omega$ was the key ingredient (for good/not good agreement). The main positive result in this vein (i.e., addressing the question of good agreement between relative/non-relative quantities) is Thm. 1. in \cite{Loveridge2017a}, which we now generalize to arbitrary locally compact $G$:
		\begin{theorem}\label{th:con1}
			Let $\R = (U_\R,\E_\R,\hir)$ be a localizable principal frame and $(\omega_n)$ a localizing sequence centered at $e \in G$. Then for any $A_\S \in B(\his)$ we have
			\begin{equation}\label{eq:ftqf}
				\lim_{n \to \infty}(\Gamma_{\omega_n}\circ \Y^\R)(A_\S) = A_\S,
			\end{equation}
			where the limit is understood as usual in the ultraweak sense.
		\end{theorem}
		
		\begin{proof}
			It is enough to check the agreement of expectation values of both sides of Eq. \eqref{eq:ftqf}. Thus take $\rho \in \S(\his)$ and calculate
			\[
			\tr[\rho \, (\Gamma_{\omega_n}\circ \Y)(A_\S)] = \tr[\rho \int_G g.A_\S \, d\mu^{\E_\R}_{\omega_n}(g) ] = 
			\int_G \tr[\rho \, (g.A_\S)] \, d\mu^{\E_\R}_{\omega_n}(g)
			\]
			
			The function $g\mapsto \tr[\rho (g.A_\S)]$ is continuous and bounded, and by Prop. \ref{prop:locseqgen} the sequence of measures $(\mu^{\E_\R}_{\omega_n})$ converges weakly to $\delta_e $, so by the portemanteau theorem we have: 
			\[\lim_{n \to \infty} \int_G \tr[\rho \, (g.A_\S)] \, d\mu^{\E_\R}_{\omega_n}(g) =  \tr[\rho \, A_\S] ; \]
			therefore the sequence of operators $ (\Gamma_{\omega_n}\circ \Y^\R)(A_\S) $ converges to $A_\S $ in the ultraweak topology.
		\end{proof}
		
		Note that if $G$ is finite and the frame is ideal, no limiting procedure is needed and the agreement is exact, i.e., we can choose  $\omega^\R = \dyad{e}$ and it holds that $(\Gamma_{\omega} \circ \Y^{\R})(A_\S) = A_\S$. Thm. \ref{th:con1} above allows us to prove the following.
		\begin{proposition}\label{prop:yeninjective}
			The relativization maps $\Y^\R: B(\his) \to B(\hisr)^G$ for localizable frames $\R$ are isometric. 
		\end{proposition}
		
		\begin{proof}
			Given $A\in B(\his)$ it is shown in \cite{Loveridge2017a} that $||\Y^{\R}(A_{\S})||\leq ||A_{\S}||$, so we only need to prove the converse inequality. By Th. \ref{th:con1} there is a localizing sequence of states $(\omega_n)$ such that for all $\rho\in\S(\his)$
			\begin{align*}
				|\tr[\rho \, A_{\S} ]|&= \lim_{n\to \infty} |\tr[\rho \,(\Gamma_{\omega_n } \circ \Y^{\R}) (A_{\S})]|=\lim_{n\to \infty} |\tr[(\rho \otimes \omega^{n}) \, \Y^{\R} (A_{\S})]|\\
				&\leq \sup_{\Omega \in \S(\hisr)} |\tr[\Omega \, \Y^{\R} (A_{\S})]|=||\Y^{\R} (A_{\S})||.
			\end{align*}
			We then have $||A_{\S}||= \sup_{\rho}|\tr[\rho \, A_{\S} ]| \leq ||\Y^{\R} (A_{\S})||$.
		\end{proof}
		
		Consider now a finite group $G$ and the frame $\E_\R$ given by
		\[
		G \ni g \mapsto P(g) = \dyad{g} \in B(L^2(G)).
		\]
		Then any state is an $\ket{e}$-product relative state since 
		\begin{equation}
			\Y^\R_*(\rho \otimes \dyad{e}) = \sum_{g \in G}\tr[P(g)\dyad{e}]g.\rho = \sum_{g \in G} \delta_{ge}g.\rho_S = \rho.
		\end{equation}
		
		This is the setting considered in \cite{de2020quantum}, and states of the form $\rho \otimes \dyad{e}$ are called ``aligned" (to the identity) \cite{krumm2021quantum}. Since $\Y^\R_*$ is constant on orbits, the above calculation is not sensitive to whether the state is aligned or only alignable \cite{krumm2021quantum}. In this work, we can make the intuition considered in \cite{giacomini2019quantum,de2020quantum} (of the frame state being `localized at the identity') rigorous in the general context of locally compact groups through the use of localizing sequences. Dualizing Th. \ref{th:con1} for a localizing sequence $(\omega_n)$ centred at $e \in G$ and an arbitrary $\rho \in \S(\his)$, we have
		\begin{equation} \label{eq:ftsr}
			\lim_{n\to \infty}\Y^\R_*(\rho \otimes \omega_n)
			=\lim_{n\to \infty} \int_G g.\rho \, d\mu^{\E_\R}_{\omega_n}(g) = \rho.
		\end{equation}
		Thus we arrive at the following proposition:
		\begin{proposition}\label{cor:locrelst}
			Let $\R$ be a localizable frame. Then $\S(\his)^\R$
			is dense in the operational topology on $\S(\his)$. 
		\end{proposition}
		
		Hence, for localizable frames, any state of $\S$ can be arbitrarily well approximated by a sequence of relative states of the form $\Y^\R_*(\rho \otimes \omega_n)$, and therefore (to arbitrary approximation) `all states are relative states' in this regime. This justifies the free use of statements of the form `the state relative to the frame is...' in other work, showing that the localizability of the frame, and the localization of the frame state is needed.  Furthermore
		the understanding that $\rho$ relative to $\dyad{e}$ takes the form $\rho \otimes \dyad{e}$ (e.g. \cite{de2020quantum}) has now been given a firm operational basis.

		We understand the story so far as follows. For localizable frames, the standard kinematics of quantum mechanics is recovered in the operational sense as a limiting procedure of localizing the state of the reference. Given that typically a large Hilbert space dimension is required for localizability, this points to a sort of classicality requirement on the frame. There is also a kind of converse to this statement proven in \cite{Miyadera2015e}: if $\R$ is not localizable, there is a lower bound on the difference between an effect $\F_\S$ and $(\Gamma_{\omega} \circ \Y^{\R})(\F_\S)$. Hence (see also \cite{Loveridge2017a} for a more detailed account):
		in the presence of a `good' (localizable) frame, the ordinary framework of quantum mechanics based on $\his$ and $B(\his)$ (with its states and observables) captures arbitrarily accurately the true, relative world described by $B(\his)^{\R}$, upon localizing the frame state $\omega$, thereby identifying $B(\his)^{\R}_{\omega}$ with $B(\his)$. In other words, the \emph{fact} that the states and (non-invariant) observables of $B(\his)$ can be used by physicists to accurately empirically describe quantum physical experiments is accounted for by the equality of probability distributions that can be generated by the relative description inside $B(\hisr)$, given a localized state of the frame. Thus the description of a quantum system given in terms of $\his$ and $B(\his)$ is understood as being relative to a well localized external quantum reference frame. That this description well approximates the true, relative one is, we suspect, is the reason that $B(\his)$ (with its states and observables) was discovered to be the correct description of a quantum system $\S$ is the first place---sufficiently good frames
		are typically used in quantum experiments. Therefore $B(\his)$ with its state space \emph{actually} describes the relation between a quantum system and a good frame. Given that, by contrast, in the setting that the frame is not localizable, or the state is not localized, there are elements of $B(\his)$ which can never capture experimentally realizable measurement statistics for that given experimental arrangement, the description afforded by $\his$ and $B(\his)$ is understood as being in the setting of a hypothetical well localized external QRF. It is this that makes the whole of $B(\his)$ operationally relevant.
		
		\section{Quantum Reference Frame Transformations}\label{sec:fcp}
		
		In this section, we provide the quantum reference frame transformations founded on the ideas developed thus far. The starting point is the frame-independent invariant algebra $B(\hi_\T)^G$ ($\hi_\T$ is the total Hilbert space of all systems under consideration), which is the `universe of discourse' for the internal quantum reference frames program as envisioned in this work.
		
		The goal is to choose some subsystem (the `initial frame') with respect to which the other subsystems can be described, and find a means by which to change the description relative to that subsystem to a description relative to another subsystem, also internal to the given setup. Therefore it is assumed that $\hi_\T$ is of the form $\hi_\T \cong \his \otimes \hio \otimes \hit$, on which we have a strongly continuous unitary representation $U_\T: G \to B(\hi_\T)$, such that $U_\T = U_{\S} \otimes U_1\otimes U_2$. This tensor factorization of spaces and actions can be understood relative to a sufficiently good external frame, but once this is understood and the symmetry imposed, no external systems need be considered.
		
		The frame-independent algebra $B(\hi_\T)^G$ contains within it the various collections of relative observables, defined by the frames $\R_1$ and $\R_2$. The observables relative to both frames $\R_1$ and $\R_2$ give the operator space needed for the frame transformations, i.e., we consider 
		\[
		B(\hi_\mathcal{T})^{\R_1,\R_2} = B(\his \otimes \hit)^{\R_1} \cap B(\his \otimes \hio)^{\R_2} \subset B(\hi_\mathcal{T})^G.
		\]
		Then, in the case of localizable frames, switching from $[\Y^{\R_1}_*(\Omega)]_{\E_2}$ to $[\Y^{\R_2}_*(\Omega)]_{\E_1}$, where $\Omega \in \S(\hi_\T)$, is a matter of representing the state $[\Omega]_{\R_1,\R_2}$ as a class of states in $\S(\his \otimes \hit)^{\R_1}$ or in $\S(\his \otimes \hit)^{\R_2}$.
		
		In this paper we analyze the quantum reference frame transformations in the case in which the initial frame $\R_1$ is localizable. A central observation of the perspective-neutral approach (see e.g.  \cite{de2021perspective}) is that to transform between relative descriptions, one must pass through a `global' description, given as the physical Hilbert space comprising the system and all frames; the framework here follows this point of view, but differs in the implementation. We will also see that the frame change we provide captures the intuition given in \cite{de2020quantum} - that to pass from the relative state to the corresponding `global' state we need to `attach the identity state'. Whilst the states $\ket{e}$ used in \cite{de2020quantum} are not available as normal states in the case of continuous groups, we can make the idea precise in the context of  locally compact groups, using localizing sequences as we have seen at various points already. Indeed, we may use Eq. \eqref{eq:ftsr}, which we recall states that for a localizing sequence $(\omega_n)$ centred at the identity $e \in G$, we have
		\begin{equation}\label{arblift}
			\lim_{n \to \infty}\Y^\R_*\left(\Omega^\R \otimes \omega_n \right) = \lim_{n \to \infty}\int_G g.\Omega^\R \, d\mu_{\omega_n}^{\E_\R}(g) = \Omega^\R,
		\end{equation}
		to see that in the case of a localizable principal frame $\R$, any $\R$-relative state can be obtained as the limit of a sequence of $\omega_n$-product $\R$-relative states as above. This is the idea behind the lifting construction, which we introduce in the next subsection. After that, we explain how operational equivalence is employed in the context of a pair of frames, defining the framed relative descriptions, to finally be able to provide the localizable frame change transformations as state space maps between the relevant operational state spaces. 
		
		\subsection{Lifting}
		
		Consider once again the Hilbert space representation $U_\S \otimes U_\R: G \to B(\hisr)$. The operators in $B(\hisr)^G$ may be conditioned upon states of the reference system by applying the restriction maps. Assuming that $\R$ is a frame with frame observable $\E_\R: \mathcal{B}(G) \to B(\hir)$, we can relativize the restricted operators, arriving at a subset of the relative ones. 
		
		Thus we introduce the map 
		\begin{equation}\label{restrel}
			\Gamma^\R_\omega := \Y^\R \circ \Gamma_\omega \circ i^G: B(\hisr)^G \to B(\his)^\R,
		\end{equation}
		where $i^G$ denotes the inclusion $B(\hisr)^G \hookrightarrow B(\hisr)$.
		As a composition of such maps $\Gamma^\R_\omega$ is unital, normal, and completely positive. It may be understood as providing a relative `version' of a given invariant. 
		
		The discrepancy between an invariant effect $\F$ and its relative version $\Gamma^\R_\omega(\F)$, given a frame prepared in the state $\omega$, is quantified as $||\Gamma^\R_\omega(\F)-\F||$, which is an operational measure since for an effect $\F \in \mathcal{E}(\hisr)^G$, $||\F||= \sup_{\rho \in \S(\hisr)}|\tr[\rho \F]|$, i.e., the discrepancy has a probabilistic interpretation.
		The best case is estimated as $\inf_{\omega}||\Gamma^\R_\omega(\F)-\F||$. If $\R$ is localizable and $\F$ is already relative, it is readily seen that this discrepancy can be made arbitrarily small.

		\begin{definition}
			The predual of $\Gamma^\R_\omega$ pre-composed with the state space isomorphism $y_\R$ under which $\S(\his)^\R \cong \S(\his)_\R$ (c.f. \ref{prop:y}) is denoted by
			\[
			\mathcal{L}^\R_\omega:= (\Gamma^\R_\omega)_* \circ y_\R: \S(\his)^{\R}  \to \S(\hisr)_G
			\]
			and called the $\R$-\emph{relative} $\omega$-\emph{lifting map} (or just `lifting map' when the context is understood).
		\end{definition}
		
		The lifting procedure allows for the `attachment' of the state of the frame to the state of the system, whilst respecting the
		symmetry-induced operational equivalence class structure. In the case of a localizable frame $\R$, using $\mathcal{L}^\R_{\omega}$ we can lift an arbitrary $\R$-relative state to a state on the invariant algebra (as the $G$-equivalence classes), which gives back the initial relative state to arbitrary precision, upon applying $\Y^\R_*$. Indeed, for $\Omega \in \S(\hisr)$ and $\Omega^\R = \Y^\R_*(\Omega)$ we have
		\[
		\mathcal{L}^\R_\omega : \S(\his)^\R \ni \Omega^\R \mapsto [\Omega^\R \otimes \omega ]_G \in \S(\hisr)_G,
		\]
		and given a localizing sequence $(\omega_n)$ centred at $e \in G$ we get (see \eqref{arblift})
		\[
		\lim_{n \to \infty} \Y^\R_* \circ \mathcal{L}^\R_{\omega_n} (\Omega^\R) = \lim_{n \to \infty} \Y^\R_*[\Omega^\R \otimes \omega_n]_G= \Omega^\R.
		\]
		Thus in the case of a localizable frame $\R$, the lifting maps taken with respect to a localizing sequence of frame states provide an approximate right inverse to the predual $\Y^\R_*$ map, which gives rise to the relative states represented as states of the system. This is in analogy to the inverse of the reduction map of the perspective-neutral approach \cite{de2021perspective}. Since $y_\R: \S(\his)^\R \xrightarrow{\sim} \S(\his)_\R$ is a state space isomorphism, which in general does not lift to an isomorphism at the level of Banach spaces (see \ref{prop:y}), we
		are only permitted to lift states (and not general trace class operators). The map \eqref{restrel} can be depicted on the following diagram
		
		\vspace{10pt}
		\begin{center}
			\begin{tikzcd}
				B(\hisr)^G \arrow[rr, "\Gamma^\R_\omega"] \arrow[rd, "\Gamma_\omega \circ i^G"] &   & B(\his)^\R \\
				& B(\his) \arrow[ru, "\Y^\R"] &          \end{tikzcd},
		\end{center}
		\vspace{10pt}
		
		with the predual, restricted to state spaces and including also  the isomorphism $y_\R$ and the lifting map, taking the~form
		
		\vspace{10pt}
		\begin{center}
			\begin{tikzcd}
				\S(\hisr)_G &                                                       & \S(\his)_\R \arrow[ll, "(\Gamma_\omega)_*"'] \arrow[ld, "\Y^\R_*"', hook'] &  & \S(\his)^\R \arrow[ll, "y_\R"'] \arrow[llll, "\mathcal{L}_\omega^\R"', bend right] \arrow[llld, "i_\R", hook'] \\
				& \S(\his) \arrow[lu, "\pi_G \circ \mathcal{V}_\omega"'] &    &  &                   
			\end{tikzcd}.
		\end{center}
		
		\subsection{Framed Relative Observables and States}
		
		In order for the operational frame transformations to be well-defined and invertible, we must make simultaneous use of relativization and framing. More precisely, we need to consider a pair of frames and restrict the domain of relativization to the operators framed with respect to the other frame. This then respects the role of both frame observables and invariance. In the context of localizable frames, such a setup turns out to be symmetric in the roles played by the frames, which gives rise to a frame transformation map analogous to those found in other works. Thus the following definition.
		
		\begin{definition}
			Given a pair of frames, $\R_1$ and $\R_2$, and a system $\S$, we will denote by
			\begin{equation}\label{relframeff}
				B(\his \otimes \hit)^{\R_1,\E_2} = {\rm span}\{\Y^{\R_1} (\F_\S \otimes \E_2(X)) \hspace{3pt} | \hspace{3pt} X \in\mathcal{B}(G), \hspace{3pt} \F_\S \in \Eff(\his)\}^{\rm cl}.
			\end{equation}
			the (ultraweakly closed) operator space of $\R_1$-\emph{relative} $\E_2$-\emph{framed operators}.
		\end{definition}
		Note that without loss of generality, the span may be taken inside the parenthesis and $\Y^{\R_1}$ (provided the image is closed at the end), to ensure that the domain of the $\R_1$ relativization map is the space of $\E_2$-framed operators. The equivalence classes of states taken with respect to such relative framed operators can be mapped to the corresponding objects with the roles of the frames interchanged, invertibly under the assumption of localizability of the frames.
		
		\begin{definition}
			The $(\R_1,\E_2)$-equivalence relation on $\T(\his \otimes \hio \otimes \hit)$, denoted $T \sim_{(\R_1,\E_2)} T'$, is defined as  the operational equivalence with respect to $B(\his \otimes \hit)^{\R_1,\E_2}$. Elements of the corresponding operational state space
			\[
			\S(\his \otimes \hit)_{\R_1,\E_2} := \S(\his \otimes \hio \otimes \hit)/\hspace{-3pt}\sim_{(\R_1,\E_2)} 
			\]
			will be called  $\E_2$-\emph{framed} $\R_1$-\emph{relative states}.
		\end{definition}
		
		We omit here the obvious definition of framed relative trace class operators and the corresponding duality result since they won't be needed. Notice the change of the order -- we apply relativization to framed operators, while we frame relative states. Crucially, the $\E_2$-framed $\R_1$-relative state space can equivalently be seen as the $\sim_{\E_2}$-quotient of the $\R_1$-relative state space:
		\begin{equation}\label{eq:rafd}
			\S(\his \otimes \hit)^{\R_1}_{\E_2} := \S(\his \otimes \hit)^{\R_1}/\hspace{-3pt}\sim_{\E_2} \cong \S(\his \otimes \hit)_{\R_1,\E_2},
		\end{equation}
		where $\S(\his \otimes \hit)^{\R_1} \subseteq \S(\his \otimes \hit)$. We then consider states in $\S(\his \otimes \hit)$ that live in the image of $\Y^{\R_1}_*$, so $R_1$-relative states on the composite system $\his \otimes \hit$, and then identify those of them that cannot be distinguished with $\R_2$-framed operators, thereby acknowledging the choice of the second frame observable. This perspective is convenient in the context of frame transformations, since it allows for the application of the lifting procedure. 
		
		The following notation will turn out to be useful. The affine quotient projection corresponding to the framing with respect to the second frame will be denoted by
		\[
		\pi^{\R_1}_{\E_2}: \S(\his \otimes \hit)^{\R_1} \twoheadrightarrow \S(\his \otimes \hit)^{\R_1}_{\E_2},
		\]
		while $[\Omega^{\R_1}]_{\E_2}$ refers to the $\E_2$-equivalence class of the $\R_1$-relative state $\Y^{\R_1}_*(\Omega) \in \S(\his \otimes \hit)$.
		
		As a side comment, notice that in the absence of the system $\S$, we have
		\[        \S(\hit)^{\R_1}_{\E_2} = \S(\hit \otimes \hio)/\hspace{-3pt}\sim_{\E_1 * \E_2},
		\]
		where $\sim_{\E_2 * \E_1}$ denotes operational equivalence with respect to the relative orientation observable $\E_2 * \E_1 = \Y^{\R_1} \circ \E_2$. Thus in this situation, the $\E_2$-framed $\R_1$-relative states only allow for the measurement of the relative orientation observable. It is easy to see that in the general case, the framed relative states also allow one to separate states with respect to the observable $\mathbb{1}_\S \otimes \E_2 * \E_1$. Indeed $\Y^\R(\mathbb{1}_\S \otimes \E_2(X)) = \mathbb{1}_\S \otimes \E_2 * \E_1 (X) \in B(\his \otimes \hit)^{\R_1,\E_2}$. We finally note that on the left hand side of Eq.\eqref{eq:rafd}, instead of framing with respect to $\E_2$, one could instead consider the smaller class of observables defined by relativising observables of $\S$ with respect to $\E_2$ and then taking the quotient on the states thus arising. Any pair of states that cannot be distinguished by $\E_2$-framed observables also cannot be distinguished by the $\R_2$-relative observables constructed by relativising those of $\S$, and therefore the approach afforded by framing, particularly the operational identification of states, also respects the point of view that only relative observables may be measured. 
		
		\subsection{Changing Reference}\label{subsec:frmchange}
		
		An operational quantum reference frame transformation takes $\E_2$-framed $\R_1$-relative states to the $\E_1$-framed $\R_2$-relative ones, i.e., it is a map
		\begin{equation}\label{frmchmap}
			\Phi_{1 \to 2}:  \S(\his \otimes \hit)^{\R_1}_{\E_2} \to \S(\his \otimes \hio)^{\R_2}_{\E_1}.
		\end{equation}
		
		In analogy to the perspective-neutral approach \cite{de2021perspective}, the frame change map is constructed by passing through the `global' description which involves the system and both frames, whilst making sure to respect the symmetry, which is given in terms of  $\S(\his \otimes \hio \otimes \hit)_G$. We first employ the lift $\mathcal{L}_{\omega}^{\R_1}$, followed by the application of $\Y^{\R_2}_*$, giving a state relative to $\R_2$. Imposing the $\E_1$ and $\E_2$ equivalences on top of that and localizing the lifting gives a consistent notion of invertible operational frame transformations for localizable frames. Indeed, we obtain the following (for brevity whenever we write `localizing sequence' without further qualification, this is centred at $e \in G$):

		\begin{definition}
			Consider a pair of principal frames $\R_1$ and $\R_2$ with $\R_1$ localizable and denote by $\E_1$ and $\E_2$ their frame observables. For a system $\S$ the map 
			\[
			\Phi^{\rm loc}_{1 \to 2}:= \lim_{n \to \infty} \pi^{\R_2}_{\E_1} \circ \Y^{\R_2}_* \circ \mathcal{L}^{\R_1}_{\omega_n}: \S(\his \otimes \hit)^{\R_1}_{\E_2} \to \S(\his \otimes \hio)^{\R_2}_{\E_1},
			\]
			where $(\omega_n)$ is any localizing sequence for $\E_1$, will be called an \emph{(internal) localized frame transformation}.
		\end{definition}
		
		Concretely, the localized frame transformation evaluated on an $\E_2$-equivalence class of $\R_1$-relative states gives
		
		\begin{equation}\label{eq:fcm}
			\Phi^{\rm loc}_{1 \to 2}: [\Omega^{\R_1}]_{\E_2} \mapsto \lim_{n \to \infty}[\Y^{\R_2}_* \circ \mathcal{L}^{\R_1}_{\omega_n}(\Omega^{\R_1})]_{\E_1} = \lim_{n \to \infty}[\Y^{\R_2}_* (\Omega^{\R_1} \otimes \omega_n)]_{\E_1}.
		\end{equation}
		
		\begin{theorem}\label{locfrtrans1}
			
			Localized frame transformations are well-defined state space maps  making the following diagram commute
			
			\vspace{10pt}
			\begin{center}
				\begin{tikzcd}
					& \mathcal{S}(\his \otimes \hio \otimes \hit)_G \arrow[ld, "\pi^{\R_1}_{\E_2} \circ \Y^{\R_1}_*"'] \arrow[rd, "\pi^{\R_2}_{\E_1} \circ \Y^{\R_2}_*"] &   \\
					\mathcal{S}(\his \otimes \hit)^{R_1}_{\E_2}  \arrow[rr,"\Phi_{1 \to 2}^{\rm loc}"] &                                                                             &  \mathcal{S}(\his \otimes \hio)^{R_2}_{\E_1}.
				\end{tikzcd}
			\end{center}
			If $\R_2$ is also localizable, the analoguosly defined map $\Phi_{2 \to 1}^{\rm loc}$ provides the inverse, i.e.,
			\[
			\Phi_{2 \to 1}^{\rm loc} \circ \Phi_{1 \to 2}^{\rm loc} = \text{\emph{Id}}_{\mathcal{S}(\his \otimes \hit)^{R_1}_{\E_2}}.
			\]
		\end{theorem}
		
		\begin{proof}
			See Appendix \ref{proof1}.
		\end{proof}

		As anticipated, the formula $\lim_{n \to \infty} \pi^{\R_2}_{\E_1} \circ \Y^{\R_2}_* \circ \mathcal{L}^{\R_1}_{\omega_n}$ gives a well-defined map $\S(\his \otimes \hit)^{\R_1} \to \S(\his \otimes \hio)^{\R_2}_{\E_1}$, so that we could consider transforming $\R_1$-relative states to $\E_1$-framed $\R_2$-relative states, although in an inevitably irreversible way. Moreover, without assuming localizability of $\R_1$ we can define an obvious non-localized but instead \emph{state-dependent} analogue of the localized frame transformations
		\begin{equation}
			\Phi_{1 \to 2}^\omega := \pi^{\R_2}_{\E_1} \circ \Y^{\R_2}_* \circ \mathcal{L}^{\R_1}_\omega: \S(\his \otimes \hit)^{\R_1} \to \S(\his \otimes \hio)^{\R_2}_{\E_1}.
		\end{equation}
		Analyzing properties of this construction is left for future work. In case $\R_2$ is not localized, the invertibility is lost, pointing to a form of decoherence.
		
		With three frames, i.e., with $\hi_\T \cong \his \otimes \hio \otimes \hit \otimes \hith$, diagonal $G$-action on $\hi_\T$ and three covariant POVMs $\E_1,\E_2$ and $\E_3$ as frame observables, the frame transformations compose in the following sense:
		
		\begin{theorem}\label{locfrtrans2}
			Assume $\R_1$ and $\R_2$ 
			to be localizable principal frames and let $\R_3$ be an arbitrary principal frame. Then 
			\begin{equation}\label{eq:come}
				\pi^{\R_3}_{\E_2} \circ \Phi^{\rm loc}_{1 \to 3} = \Phi^{\rm loc}_{2 \to 3} \circ \Phi^{\rm loc}_{1 \to 2},
			\end{equation}
			where $\pi^{\R_3}_{\E_2}$ denotes the $\E_2$-framing projection from $\S(\his \otimes \hio \otimes \hit)^{\R_3}_{\E_1}$.
		\end{theorem}
		The map $\pi^{\R_3}_{\E_2}$ is required in order to account for the choice of the second frame observable.
		
		\begin{proof}
			See Appendix \ref{proof2}.
		\end{proof}
		
		\subsection{Comparison to other work}\label{subsec:comparison}
		It is worth making a brief comparison to the works \cite{giacomini2019quantum,de2020quantum} (which have since been dubbed `purely perspectival' approaches), and the perspective-neutral approach (e.g. \cite{de2021perspective}). We begin with the former. For convenient comparison, we employ here the convention adapted in other works and write the Hilbert spaces (and states) of the frames on the left side of the tensor product.
		
		In \cite{giacomini2019quantum} three copies of $L^2(\mathbb{R})$ are considered, with the compound system given as $L^2(\mathbb{R})_{\R_1}\otimes L^2(\mathbb{R})_{\R_2} \otimes L^2(\mathbb{R})_{\R_3}$; we will denote the factors by $\hi_{\R_i}$ to simplify notation. A unitary map
		\[
		U := e^{i Q_{\R_1} \otimes P_{\R_2}}: \hi_{\R_1}\otimes \hi_{\R_2} \to \hi_{\R_1}\otimes \hi_{\R_2}
		\]
		is introduced, giving $(U \varphi_1 \otimes \phi_2)(x,y)= \varphi_1(x)\phi_2(y+ x)$;
		this is composed with the isometric bijection $V: \hi_{\R_1} \to  \hi_{\R_3}$ defined by $(V \varphi_1) (x) = \varphi_1 (-x)=:\varphi_3(x)$, and the $SWAP_{1,3}$ exchanging the order of tensor factors. The frame change map of \cite{giacomini2019quantum} is then the composition
		\[
		F_{3 \to 1}:= SWAP_{1,3} \circ (V \otimes \id_{\R_2}) \circ U: \hi_{\R_1} \otimes \hi_{\R_2} \to \hi_{\R_2} \otimes \hi_{\R_3},
		\]
		which on product states gives
		\begin{equation}\label{eq:fc1}
			\hi_{\R_1} \otimes \hi_{\R_2} \ni \varphi_1(x) \phi_2(y) \mapsto \phi_2(y+x)\varphi_3(-x) \in \hi_{\R_2} \otimes \hi_{\R_3}.
		\end{equation}
		The interpretation given is that the state on the left hand side is the state relative to $\mathcal{R}_3$, which is always assumed to be perfectly localized (though not part of the Hilbert space description) and the right hand side the transformed state, relative to $\R_1$, which again is not part of the Hilbert space description, but is assumed to be perfectly localized. We mention again that motivation for the choice of localized frame state is not given in \cite{giacomini2019quantum}.  
		Generically, the right hand side of \eqref{eq:fc1} is entangled; indeed, it appears that the only setting in which no entanglement is `generated' through the frame change is for a product of perfectly localized `position eigenstates'. The work of \cite{de2020quantum} explicitly includes the state of $\R_1$ as the $\ket{0}$ in the theoretical description (later generalised to $\ket{e}; e \in G$). This is stated to be a conventional choice in \cite{de2020quantum}.
		As we have commented, the perfectly position-localized states are not part of the Hilbert-space framework of quantum theory, and nor are the $g$-localized states for $g \in G$ for typical (locally compact) $G$, which we have taken care to avoid, replacing such objects with localizing sequences. The use of such sequences highlights another important aspect missing from the above works: the sequences give rise to localized probability measures with respect to the frame observable; in the perspectival approaches the observables and the probability distributions they give rise to are not considered, and therefore only half the picture is presented.
		To connect to our work, we note that $\mathcal{R}_2$ above is our $\S$, 
		and $\mathcal{R}_3$ is our $\mathcal{R}_1$ (the initial frame). 
		
		It is important to compare the probability distributions arising in \cite{giacomini2019quantum,de2020quantum} with those arising in this work in order to make a concrete connection to the framework we have provided, in which we have introduced a symmetry principle on the observables, combined with an operational justification of the (more general, here) notion of relative state.
		
		Since the perfectly localized states are possible only for discrete $G$, we fix a finite $G$
		and change convention for the $G$-actions/representations on the frames to be given by 
		$U(g)\ket{h}:=\ket{hg^{-1}}$ (noting that on $G$ this is a left action). Such states can be used to construct a frame observable $P(g)=\dyad{g^{-1}}$, which is easily seen to be a covariant PVM. We set $\his \cong \hit \cong \hit \cong l^2(G)$.\footnote{Note that each frame observable $P_i$ generates a commutative (`classical') subalgebra through $\{P(g)\}'' \cong diag(M_n(\mathbb{C})) \cong C(G) \subset B(L^2(G)),$
			where $'$ denotes commutant and $C(G)$ is the set of functions $G\to \mathbb{C}$ equipped with pointwise multiplication (noting the distinction between $C(G)$ and $L^2(G)$ which we view only as a linear space). The pure states of $C(G)$ are in bijection with elements of $G$ and therefore also the frame projections $P(g)$.} To undertake the comparison, we begin with the classical situation, which motivates the frame change prescriptions in \cite{giacomini2019quantum,de2020quantum}. At this level the localized frame change transformation given here, when understood as a map between the \emph{classical} pure states, identified with the projections on the $G$-basis, yields (noting that the relevant classes are trivial in this case)
		\begin{align*}
			\left( \dyad{h}_{\R_2} \otimes \dyad{g}_{\S}\right)^{\dyad{e}_{\R_1}} &\mapsto 
			\dyad{e}_{\R_1}\otimes \dyad{h}_{\R_2}\otimes \dyad{g}_{\S}\\
			 &\mapsto \left(\dyad{h^{-1}}_{\R_1} \otimes \dyad{gh^{-1}}_\S\right)^{\dyad{e}_{\R_2}},
		\end{align*}
		which is identical to that given in \cite{de2020quantum}. In \cite{de2020quantum}, this map is unitarily extended at the Hilbert space level (all phases are set~to~$1$) by the ``principle of coherent change of reference system" (which was assumed also implicitly in \cite{giacomini2019quantum}). To see how the construction here and that of \cite{de2020quantum} differs on the quantum level, consider $\omega \in \S(\hit)$ and $\rho \in \S(\his)$. Writing $U_{1 \to 2}$ for the frame change unitary corresponding to that given in \cite{de2020quantum}, we find
		\[
		U_{1 \to 2}(\omega \otimes \rho)U^*_{1 \to 2} = \sum_{g,h} \bra{g} \omega \ket{h} \dyad{g^{-1}}{h^{-1}} \otimes U_\S (g) \rho U^*_\S (h).
		\]
		On the other hand, the localized frame change transformation given in \eqref{eq:fcm}, adapted to this setting, gives
		\[
		\Phi_{1 \to 2}^{\rm loc} (\omega \otimes \rho) = \sum_{g} \bra{g} \omega \ket{g} \left[\dyad{g^{-1}} \otimes U_\S (g) \rho U^*_\S (g)\right]_{\E_1}.
		\]
		Perhaps surprisingly, it turns out that the two states resulting from these procedures are operationally equivalent, i.e,
		\[
		U_{1 \to 2}(\omega \otimes \rho)U^*_{1 \to 2} \in \Phi_{1 \to 2}^{\rm loc} (\omega \otimes \rho).
		\]
		Our procedure is then perfectly compatible with the transformations of \cite{de2020quantum} whenever the latter is rigorously defined, which we now state precisely:
		
		\begin{proposition}\label{prop:compACT}
			Consider a finite group $G$ and a pair of ideal frames $\R_1$, $\R_2$. Then,
			\[
			\Phi_{1 \to 2}^{\rm loc} = \pi^{\R_1}_{\E_2} \circ U_{1 \to 2}(\_)U^*_{1 \to 2}: \S(\his \otimes \hio)_{\E_2}^{\R_1} \to  \S(\his \otimes \hit)^{\R_2}_{\E_1}.
			\]
		\end{proposition}
		
		\begin{proof}
			We calculate
			\begin{align*}
				&\tr[U_{1 \to 2}\Omega^{\R_1}U^*_{1 \to 2}\dyad{l^{-1}}\otimes \F_\S] \\
				&=\tr[\sum_{g,h}\dyad{g^{-1}}{g} \otimes U_\S(g)\Omega^{\R_1}\dyad{h}{h^{-1}} \otimes U^*_\S(h)\dyad{l^{-1}} \otimes \F_\S] \\
				&=\tr[\sum_g\dyad{g^{-1}}{g} \otimes U_\S(g)\Omega^{\R_1}\dyad{l}{l^{-1}} \otimes U^*_\S(g)\F_\S] \\
				&=\tr[\sum_g\Omega^{\R_1}\dyad{l}{l^{-1}} \otimes U^*_\S(g)\F_\S\dyad{g^{-1}}{g} \otimes U_\S(g)] \\
				&=\tr[\Omega^{\R_1}\dyad{l} \otimes U^*_\S(l)\F_\S U_\S(l)] =\tr[\Omega^{\R_1}\dyad{l} \otimes l^{-1}.\F_\S] \\
			\end{align*}
			By contrast,
			\begin{align*}
				&\tr[\Phi^{\rm loc}_{1 \to 2}(\Omega^{\R_1})\dyad{l^{-1}}\otimes \F_\S] \\
				&=\tr[\dyad{e} \otimes \Omega^{\R_1} \sum_g \dyad{g^{-1}} \otimes g. \left(\dyad{l^{-1}} \otimes \F_\S\right)]\\
				&=\tr[\dyad{e} \otimes \Omega^{\R_1} \sum_g \dyad{g^{-1}} \otimes \dyad{l^{-1}g^{-1}} \otimes g.\F_\S]\\
				&=\tr[\Omega^{\R_1} \sum_g \dyad{g^{-1}} \delta(l,g^{-1}) \otimes g.\F_\S]=\tr[\Omega^{\R_1} \dyad{l} \otimes l^{-1}.\F_\S].
			\end{align*}
		\end{proof}
		In particular, this shows that if the frame $\R_2$ is prepared in a superposed state, e.g. the input vector of the frame change~is
		\begin{equation}
			\ket{\psi}=\left((\alpha \ket{h_1}_{\R_2}+\beta \ket{h_2}_{\R_2}) \otimes \ket{g}_\S\right)^{\ket{e}_{\R_1}},
		\end{equation}
		at an operational level there is no difference between the transformed state as given in \cite{de2020quantum}
		\begin{equation}\label{eq:fecc1}
			\left( \alpha \ket{h_1^{-1}}_{\R_1} \otimes \ket{gh_1^{-1}}_\S  +\beta \ket{h_2^{-1}}_{\R_1} \otimes \ket{gh_2^{-1}}_\S \right)^{\ket{e}_{\R_2}},
		\end{equation}
		and that arising from our procedure on a representative of the $\E_2$-equivalence class of $\dyad{\psi}$, which reads
		\begin{equation}\label{eq:fecc}
			\scalebox{0.8}{$\left(|\alpha|^2 \dyad{h_1^{-1}}_{\R_1} \otimes \dyad{gh_1^{-1}}_\S 
				+|\beta|^2 \dyad{h_2^{-1}}_{\R_1} \otimes \dyad{gh_2^{-1}}_\S \right)^{\dyad{e}_{\R_2}},$}
		\end{equation}
		since they occupy the same $\E_1$-equivalence class in $\S(\hit \otimes \his)^{\R_2}$. The state in  \eqref{eq:fecc} is the L\"{u}ders mixture corresponding to \eqref{eq:fecc1}, and is \emph{not} entangled. (Note that the same conclusion---that \eqref{eq:fecc1} and \eqref{eq:fecc} are equivalent---can be drawn by evaluating any invariant observable in \eqref{eq:fecc1} and noting that the cross-terms vanish). It may be tempting from Eq. \eqref{eq:fecc1} (and \cite{giacomini2019quantum}) to draw strong physical conclusions that ``superposition and entanglement are frame-dependent''. Whilst we do not necessarily disagree with the broad understanding of quantum properties depending on frame choices---indeed we have seen such behaviour here---the precise notion that a superposition state for $\R_2$ (tensored with any state of $\S$) is transformed into an entangled state of $\S$ and $\R_1$, we believe deserves further scrutiny; as we can see it is not as innocuous as one might think from an operational perspective.

		We finish with a mention of the perspective-neutral framework (e.g. \cite{Vanrietvelde:2018pgb,ahmad2022quantum,krumm2021quantum,de2021perspective}) and how it relates conceptually and technically to the construction presented here. The perspective-neutral approach 
		plays out on the physical Hilbert space $\hi_{\rm phys}$, which is defined, for compact $G$, by one of three equivalent procedures: (i) as the kernel of a quantised classical constraint which is defined in some kinematical Hilbert space $\hi_{\rm kin}$, (ii) as the (closed linear span of) the invariant vectors under the unitary group defined through the exponentiated constraint operators, (iii) as the image of the `coherent group averaging' $\Pi: \hi_{\rm kin} \to \hi_{\rm phys}; \ket{\psi} \mapsto  = \int_G   \, d\mu(g) U(g)$, where $U$ is the relevant unitary representation of $G$. The algebra of (bounded) physical observables is then $B(\hi_{\rm phys})$, which in general does not coincide with our $B(\hi)^G$, even if the same kinematical Hilbert space is chosen. The differences are further magnified when $G$ is not compact; this does not have much effect on the operational approach, but the setting of standard Hilbert space quantum theory is no longer viable for the perspective-neutral approach, and distributional techniques are required. These make use of a  
		`rigging' construction \cite{giulini1999generality,giulini1999uniqueness}, which is not known to well defined in general, and the distributional framework is used informally in \cite{de2021perspective}.
		
		An exhaustive comparison of the technical and conceptual differences between the operational approach constructed here and the perspective-neutral approach is a serious undertaking in its own right, given that there is significant technical work to be done to ensure the well-definedness of all the objects involved. However, the frame change map can be constructed informally; at the level of rigour of that work the map is unitary even for non-ideal frames and agrees with, and therefore in a sense subsumes, \cite{giacomini2019quantum,de2020quantum} for ideal frames. 
		
		The ``relational Schr{\"o}dinger picture'' frame change map of the perspective-neutral framework can be written (setting $g_i=g_j=e$) in the following form (see Thm. 4 on pg. 40 of \cite{de2021perspective})
		\[
		V_{1\to 2} = \int_G \dyad{\phi(g)}{\psi(g)} \otimes U_\S(g) \, d\mu(g),
		\]
		where $\{\ket{\phi(g)}\}_{g \in G} \subset \hit$ and $\{\ket{\psi(g)}\}_{g \in G} \subset \hit$ are systems of coherent states of the first and second frame, respectively, understood in the sense of \eqref{ex:csspovm}, and $\dyad{\phi(g)}{\psi(g)}$ as a map from $\hit$ to $\hit$. Within the Hilbert space framework of quantum mechanics (i.e., without considering rigged Hilbert spaces and distributions), ideal coherent state frames may only be defined on discrete groups. In this setting, for $\R_1$ ideal, a similar calculation shows that there is no operational distinction between the frame change map we have provided and that of the perspective-neutral approach (see \cite{JGthesis}). 
		
		A comprehensive interpretation of the distinctions between the various frameworks is a topic of ongoing investigation. Each has different assumptions, applications and domains of validity. Important questions are yet to be answered; for instance, both the perspective-neutral (e.g. \cite{Hoehn:2019owq}) and operational frameworks \cite{loveridge2019relative} have been applied to the famed Page-Wootters mechanism \cite{page1983evolution}, in which a globally time-translation-invariant state/observable may give rise to 
		Schr\"{o}dinger/Heisenberg dynamics for one subsystem relative to another. The different frameworks differ in their implementation, and the physical content of the distinction is yet to be fully understood.
		
		\section{Conclusion}
		
		In this work, we have provided a mathematical foundation for an operationally motivated framework for quantum reference frames and their transformations.  The quantum reference frames were defined as systems of covariance based on $G$-spaces. The relativization map was used to construct relative observables, which are observables on the composite systems satisfying the framing and invariance conditions that we require from observable quantities. The relative states are defined as operational equivalence classes of states that cannot be distinguished by the relative observables, or, equivalently, as the states of the system lying in the image of the predual of the relativization map. The Banach duality between states and observables in the orthodox, non-relational quantum theory is seen to manifest on the relational level. We analysed conditioning the relative descriptions on a state of the reference, finding that in the limit of highly localized states, the relative descriptions probabilistically reproduce the standard ones, effectively externalizing the reference. All this is achieved by introducing the notion of operational equivalence and merging it with the ideas and results previously given in \cite{Loveridge2011, loveridge2012quantum, loveridge2017relativity, Loveridge2017a, Loveridge2020a}.
		
		A further contribution in this work is the provision of the invertible and composable frame transformations for internal localizable principal quantum reference frames. The frame change procedure relies upon the specification of two frames from the outset, which are understood as fixing the conditions under which the observed phenomena are to be understood. In the setting of sharp frames, this amounts to choosing particular commutative subalalgebras of the frame algebras, with the interpretation that quantum reference frames are classical windows into the quantum world. This view is morally in line with the Bohrification programme of Landsman \cite{landsman2017foundations}, but distinct from it in that only one part of a composite system is `Bohrified' and the covariance requirement strongly restricts the choice of frame. Further work is needed to cement this connection, but it appears to be an important feature of the operational approach.
		
		The frame change is operationally indistinguishable, on the common domain of rigorous applicability, from those of \cite{de2020quantum} and \cite{de2021perspective}. This relies fundamentally on the stipulation that only relative observables are truly measurable, and that the fixing of frames fixes a classical context, and in particular suggests that there is no `fact of the matter' about whether the frame change is coherent, or entangling, or not. We conclude that the strong physical claims made based on previously studied quantum reference frame transformations should therefore be treated with caution.
		
		There remains, as ever, work to be done; though the relativization map has now been given on homogeneous spaces in fairly general terms \cite{glowacki2023quantum,fewster2024quantum}, the frame changes have not yet been analysed in that setting. Topics which have been investigated within the other frameworks should be investigated operationally, for example the `relativity of subsystems' \cite{ahmad2022quantum,castro2021relative} should be systematically studied in the operational framework. More generally, it is not yet known what the precise differences in the frameworks are in terms of the physics they describe, and a thorough analysis must be done. Regarding further work from within the operational approach, the convex flavour is well suited to general probabilistic theories. The framework is clearly also amenable to a von Neumann algebraic generalisation, aspects of the theory of which have appeared in \cite{fewster2024quantum}. There, it was shown that unsharp quantum reference frames are needed for the reduction of a type $III$ to a type $II_1$ von Neumann algebra in quantum field theory
		(see also \cite{chandrasekaran2023algebra,witten2022gravity,witten2024background,de2024gravitational,de2024crossed,ali2024crossed,ali2024quantum,ahmad2024semifinite,ahmad2024relational,gomez2022cosmology,gomez2023entanglement,klinger2024crossed,jensen2023generalized,witten2024algebras,cirafici2024fluctuation} for further work on type reduction, von Neumann algebras, and quantum reference frames etc.), and further work would involve understanding the role of operational frame changes in such a setting.

		\section*{Acknowledgements}
		Thanks are due to Chris Fewster, James Waldron, Anne-Catherine de la Hamette, Stefan Ludescher, Philipp H\"{o}hn, Markus M\"{u}ller, Tom Galley, Sebastiano Nicolussi Golo, Isha Kotecha, Josh Kirklin, Fabio Mele, Jaros{\l}aw Korbicz, Marek Ku{\'s},  Alexander and Inna Boritchev, Hamed Mohammady, Takayuki Miyadera, Kjetil B{\o}rkje, and Klaas Landsman for helpful interactions. We also owe our gratitude to an anonymous referee who offered some pertinent comments and criticisms, and which certainly improved the manuscript. LL would like to thank the Theoretical Visiting Sciences Programme at the Okinawa Institute of Science and Technology (OIST) for enabling his visit, and for the generous hospitality and excellent working conditions during his time there, which significantly aided the development of this work. JG acknowledges the funding received via NCN through the OPUS grant nr.~$2017/27/$B/ST$2/02959$ and support by the Digital Horizon Europe project FoQaCiA, Foundations of quantum computational advantage, GA No. 101070558, funded by the European Union, NSERC (Canada), and UKRI (UK). This publication was made possible through the support of the ID\# 62312 grant from the John Templeton Foundation, as part of the \href{https://www.templeton.org/grant/the-quantuminformation-structure-ofspacetime-qiss-second-phase}{``The Quantum Information Structure of Spacetime'' Project (QISS)}. The opinions expressed in this project/publication are those of the author(s) and do not necessarily reflect the views of the John Templeton Foundation.

		
		
		\bibliographystyle{quantum}
		\bibliography{QRFarticleQuantumTemplate.bib}

		\appendix
		\section{Glossary of Spaces}\label{glossary}
		We collect below, for reference, the key spaces which have been used and some of the main results. All closures are understood to be with respect to the ultraweak topology.
		
		\begin{itemize}
			\item \emph{Invariant operators/effects}: von Neumann subalgebra $B(\hi)$/(convex) subset of $\Eff(\hi)$ consisting of operators/effects invariant under the given (strongly continuous) unitary representation of a (locally compact) group~$G$
			\begin{align*}
				B(\hi)^G &:= \{A \in B(\hi)\hspace{3pt}|\hspace{3pt}g.A \equiv U(g)AU^*(g)=A\} \subseteq B(\hi), \\
				\Eff(\hi)^G &:= \{\F \in \Eff(\hi)\hspace{3pt}|\hspace{3pt}g.\F \equiv U(g)\F U^*(g)=\F\} \subseteq \Eff(\hi).
			\end{align*}
			
			\item \emph{Invariant states/trace class operators}: (total convex) subset/subspace of states/trace class operators which are invariant under the given unitary representation of $G$
			\begin{align*}
				\T(\hi)^G &:= \{T \in \T(\hi)\hspace{3pt}|\hspace{3pt} T.g \equiv U^*(g)T U(g)=T\} \subseteq \T(\hi), \\
				\S(\hi)^G &:= \{\Omega \in \S(\hi)\hspace{3pt}|\hspace{3pt} \Omega.g \equiv U^*(g)\Omega U(g)=\Omega\} \subset \T(\hi)^G.
			\end{align*}
			
			\item $G$-\emph{equivalent states/trace class operators}: (total convex) operational quotient space of classes of states/trace class operators that can not be distinguished by the invariant effects (or, equivalently, by the invariant operators)
			\begin{align*}
				\T(\hi)_G &:= \T(\hi)/\hspace{-3pt}\sim_G, \\
				\S(\hi)_G &:= \S(\hi)/\hspace{-3pt}\sim_G,
			\end{align*}
			where $T \sim_G T' \hspace{5pt} \Leftrightarrow \hspace{5pt} \tr[T \F]=\tr[T' \F] \text{ for all } \F \in \Eff(\hi)^G$. In the case of compact $G$ we have
			\begin{align*}
				\T(\hi)_G &\cong \T(\hi)^G \text{ (as Banach spaces),} \\
				\S(\hi)_G &\cong \S(\hi)^G \text{ (as state spaces).}
			\end{align*}
			
			\item \emph{Framed operators/effects}: (ultraweakly closed) operator space in $B(\hisr)$/(convex) subset of $\Eff(\hisr)$ consisting of operators/effects respecting the choice of the frame-orientation observable $\E_\R: \mathcal{B}(G) \to \Eff(\hir)$
			\begin{align*}
				\Eff(\hisr)^{\E_\R} &:= \rm{conv}\left\{\F_\S \otimes \E_\R(X) \hspace{3pt} | \hspace{3pt} X \in \mathcal{B}(G), \F_\S \in \Eff(\his)\right\}^{\rm cl},\\
				B(\hisr)^{\E_\R} &:= \rm{span}\left\{\F_\S \otimes \E_\R(X) \hspace{3pt} | \hspace{3pt} X \in \mathcal{B}(G), \F_\S \in \Eff(\his)\right\}^{\rm cl}.
			\end{align*}
			\item \emph{Framed states/trace class operators}: (total convex) operational quotient space of classes of states/trace class operators that cannot be distinguished by the $\E_\R$-framed effects (or, equivalently, by the $\E_\R$-framed operators)
			\begin{align*}
				\T(\his)_{\E_\R} &:= \T(\hisr)/\hspace{-3pt}\sim_{\E_\R}, \\
				\S(\hi)_{\E_\R} &:= \S(\hisr)/\hspace{-3pt}\sim_{\E_\R},
			\end{align*}
			where $T \sim_{\E_\R} T' \hspace{5pt} \Leftrightarrow \hspace{5pt} \tr[T \, \F]=\tr[T' \, \F] \text{ for all } \F \in \Eff(\hisr)^{\E_\R}$.
			
			\item \emph{Relative operators/effects}: (ultraweakly closed) operator space in $B(\hisr)$/(convex) subset of $\Eff(\hisr)$ consisting of operators/effects lying in the ultraweak closure of the $\R$-relativization map
			\[
			\Y^\R (A_\S) := \int_G g.A_\S \otimes d\E_\R(g),
			\]
			so that we have
			\begin{align*}
				\Eff(\his)^\R &:=  \Y^\R(\Eff(\his))^{\rm cl},\\ 
				B(\his)^\R &:= 
				\Y^\R(B(\his))^{\rm cl}. 
			\end{align*}
			We have $\Eff(\his)^\R \subseteq \Eff(\hisr)^G$ and $B(\his)^\R \subseteq B(\hisr)^G$. When $\E_\R$ is sharp $B(\hir)^\R$ is a von Neumann subalgebra of $B(\hisr)^G$.
			
			\item \emph{Relative states/trace class operators}: (total convex) operational quotient space of classes of states/trace class operators that can not be distinguished by the $\R$-relative effects (or, equivalently, by the $\R$-relative operators)
			\begin{align*}
				\T(\his)_\R &:= \T(\hisr)/\hspace{-3pt}\sim_\R, \\
				\S(\hi)_\R &:= \S(\hisr)/\hspace{-3pt}\sim_\R,
			\end{align*}
			where $T \sim_\R T' \hspace{5pt} \Leftrightarrow \hspace{5pt} \tr[T \F]=\tr[T' \F] \text{ for all } \F \in \Eff(\his)^\R$. Writing
			\[
			\S(\hi)^\R := \Y^\R_*(\S(\hisr)) = \Y^\R_*(\S(\hisr)/\hspace{-3pt}\sim_G) \subseteq \S(\his),
			\]
			we have an isomorphism of state spaces
			\[
			y_\R: S(\his)^\R \ni \Y^\R_*(\Omega) \mapsto [\Omega]_\R \in \S(\his)_\R.
			\]
			We use the following notation for $\R$-relative states
			\[
			\Omega^\R \equiv \Y^\R_*(\Omega) \cong [\Omega]_\R,
			\]
			where $\Omega$ is a state on the composite system $\Omega \in \S(\hisr)$ and $[\_]_\R$ denotes the operational equivalence class taken with respect to $\Eff(\his)^\R$. When $\R$ is localizable the inclusion $\S(\his)^\R \subseteq \S(\his)$ is dense in the operational topology.
			
			\item $\omega$-\emph{conditioned relative operators/effects}: (ultraweakly closed) operator space in $B(\his)$/(convex) subset of $\Eff(\his)$ consisting of $\omega$-conditioned $\R$-relative operators/effects
			\begin{align*}
				\Eff(\his)^\R_\omega &:= \Y^\R_\omega \left(\Eff(\his)\right)^{\rm cl},\\
				B(\his)^\R_\omega &:=\Y^\R_\omega \left(B(\his)\right)^{\rm cl}, 
			\end{align*}
			where $\Y^\R_\omega := \Gamma_\omega \circ \Y^\R: B(\his) \to B(\his)$ with
			\[
			\Gamma_\omega: B(\hisr) \ni A_\S \otimes A_\R \mapsto \tr[\omega A_\R] A_\S \in B(\his)
			\]
			extended by linearity and continuity. The $\omega$-conditioned $\R$-relative operators take the form
			\[
			\Y^\R_\omega(A_\S) = \int_G U(g)A_\S U^*(g) \, d\mu^{\E_\R}_\omega(g).
			\]
			\item $\omega$-\emph{product relative states/trace class operators}: (total convex) operational quotient space of classes of states/trace class operators that can not be distinguished by the $\omega$-conditioned $\R$-relative effects (or, equivalently, by the $\omega$-conditioned $\R$-relative operators)
			\begin{align*}
				\T(\his)^\omega_\R &:= \T(\his)/\hspace{-3pt}\sim_{(\R,\omega)},\\ 
				\S(\hi)^\omega_\R &:= \S(\his)/\hspace{-3pt}\sim_{(\R,\omega)},\\
			\end{align*}
			where $T \sim_{(\R,\omega)} T' \hspace{5pt} \Leftrightarrow \hspace{5pt} \tr[T \F]=\tr[T' \F_\S] \text{ for all } \F_\S \in \Eff(\his)_\omega^\R$. We have 
			\[
			\S(\hi)^\omega_\R \cong (\Y^\R_\omega)_*(\S(\his)),
			\]
			with the $\omega$-product $\R$-relative states taking the form
			\[
			\rho^{(\omega)} := (\Y^\R_\omega)_*(\rho) = \int_G U(g)^*\rho U(g) \, d\mu^{\E_\R}_\omega(g).
			\]
			\item $\omega$-\emph{lifted relative states}: $G$-equivalent states in $\S(\hisr)_G$ that arise by attaching a frame state $\omega$ to an $\R$-relative state via the $\omega$-lifting map $\mathcal{L}^\R_\omega$ defined as
			\[
			\mathcal{L}^\R_\omega := (\Gamma^\R_\omega)_* \circ y_\R: \S(\his)^\R \to \S(\hisr)_G,
			\]
			where
			\[
			\Gamma^\R_\omega := \Y^\R \circ \Gamma_\omega \circ i^G: B(\hisr)^G \to B(\his)^\R
			\]
			is the $\R$-\emph{relative} $\omega$-\emph{restriction~map}. We then have
			\[
			\mathcal{L}^\R_\omega: \Omega^\R \mapsto [\Omega^\R \otimes \omega]_G.
			\]
			
			\item \emph{Relative framed operators/effects}: (ultraweakly closed) operator space in $B(\hio \otimes \hit \otimes \his)^G$/(convex) subset of $\Eff(\his \otimes \hio \otimes \hit)^G$ consisting of $\R_1$-relative $\E_2$-framed effects, e.g. 
			\begin{align*}
				\Eff(\his \otimes \hit)^{\R_1,\E_2} &:= \Y^{\R_1}\left(\Eff(\his \otimes \hit)^{\E_2}\right)^{\rm cl} \\
				B(\his \otimes \hit)^{\R_1,\E_2} &:= \Y^{\R_1}\left(B(\his \otimes \hit)^{\E_2}\right)^{\rm cl}.
			\end{align*}
			
			\item \emph{Framed relative states/trace class operators}: (total convex) operational quotient space of classes of states/trace class operators that can not be distinguished by the $\R_1$-relative $\E_2$-framed effects (or, equivalently, by the $\R_1$-relative $\E_2$-framed operators)
			\begin{align*}
				\T(\his \otimes \hit)^{\R_1,\E_2} &:= \T(\his \otimes \hio \otimes \hit)/\hspace{-3pt}\sim_{(\R_1,\E_2)}, \\
				\S(\his \otimes \hit)^{\R_1,\E_2} &:= \S(\his \otimes \hio \otimes \hit)/\hspace{-3pt}\sim_{(\R_1,\E_2)},
			\end{align*}
			where $T \sim_{(\R_1,\E_2)} T' \hspace{5pt} \Leftrightarrow \hspace{5pt} \tr[T \, \F]=\tr[T' \, \F] \text{ for all } \F \in \Eff(\his \otimes \hit)^{\R_1,\E_2}$. We also have
			\[
			\S(\his \otimes \hit)^{\R_1}_{\E_2} := \S(\his \otimes \hit)^{\R_1}/\hspace{-3pt}\sim_{\E_2} \cong \S(\his \otimes \hit)^{\R_1,\E_2}.
			\]
			For the corresponding quotient projections we write
			\[
			\pi^{\R_1}_{\E_2}: \S(\his \otimes \hit) \supseteq \S(\his \otimes \hit)^{\R_1} \to \S(\his \otimes \hit)^{\R_1}/\hspace{-3pt}\sim_{\E_2}.
			\]
			
			\item \emph{Localized frame transformations}: given a pair of frames $\R_1$ and $\R_2$ and a system $\S$, assuming localizibility of $\R_1$ we have the localized frame transformation, which is a state space map 
			\[
			\Phi^{\rm loc}_{1 \to 2}: \S(\his \otimes \hit)^{\R_1}_{\E_2} \to \S(\hi_{\R_1} \otimes \his)^{\R_2}_{\E_1},
			\]
			given by
			\begin{align*}
				\Phi^{\rm loc}_{1 \to 2}:= \lim_{n \to \infty} \pi^{\R_2}_{\E_2} \circ \Y^{\R_2}_* \circ \mathcal{L}^{\R_1}_{\omega_n}: [\Omega^{\R_1}]_{\E_2} &\mapsto \lim_{n \to \infty}[\Y^{\R_2}_* \circ \mathcal{L}^{\R_1}_{\omega_n}(\Omega^{\R_1})]_{\E_1}\\
				 &= \lim_{n \to \infty}[\Y^{\R_2}_* (\Omega^{\R_1} \otimes \omega_n)]_{\E_1}.
			\end{align*}
			If $\R_2$ is also localizable, the corresponding localized frame transformation provides the inverse. In the setup of three frames localizable frame transformations compose as follows
			\[
			\pi^{\R_3}_{\E_2}\circ \Phi ^{\rm loc}_{1 \to 3} = \Phi ^{\rm loc}_{2 \to 3} \circ \Phi ^{\rm loc}_{1 \to 2}.
			\]
		\end{itemize}
		
		\section{Diagrams}\label{diagrams}
		As a graphical summary of the main ingredients of relational quantum kinematics, we present a dual pair of commutative diagrams representing all the relevant maps. The $\omega \in \S(\hir)$ is as usual arbitrary, the maps $i$ denote the obvious inclusions, and the maps $\pi$ are the corresponding quotient projections. We restrict to state spaces (as opposed to the full predual spaces) to include the $\omega$-lifting map.

		\vspace{10pt}
		\begin{center}
			\adjustbox{scale=0.9,center}{%
			\begin{tikzcd}
				&    &   & B(\hisr) &   \\
				B(\hisr)^G \arrow[rr, "\Gamma^\R_\omega"] \arrow[rd, "\Gamma_\omega \circ i^G"] &   & B(\his)^\R \arrow[rd, "\Gamma_\omega \circ i^\R"] \arrow[rr, "i^G_\R", hook] \arrow[ru, "i^\R", hook] &    & B(\hisr)^G \arrow[lu, "i^G"', hook'] \arrow[ld, "\Gamma_\omega \circ i^G"] \\
				& B(\his) \arrow[ru, "\Y^\R"] \arrow[rr, "\Y^\R_\omega"] &  & B(\his)^\R_\omega &          \end{tikzcd}
			}
		\end{center}
		
		\vspace{10pt}
		\begin{center}
			\adjustbox{scale=0.9,center}{%
			\begin{tikzcd}
				&     &   & \S(\hisr) \arrow[ld, "\pi_\R"', two heads] \arrow[rd, "\pi_G", two heads] &   \\
				\S(\hisr)_G &  & \S(\his)^\R \cong S(\his)_\R \arrow[ll, "\mathcal{L}^\R_\omega"'] \arrow[ld, "i_\R \cong \Y^\R_*"', hook'] &    & \S(\hisr)_G \arrow[ll, "\pi_G^\R"', two heads] \\
				& \S(\his) \arrow[lu, "\pi_G \circ \mathcal{V}_\omega"'] &  & \S(\his)_\R^\omega \arrow[lu, "\pi_\R \circ \mathcal{V}_\omega"'] \arrow[ll, "(\Y^\R_\omega)_*"'] \arrow[ru, "\pi_G \circ \mathcal{V}_\omega"]&     
			\end{tikzcd}
		}
		\end{center}
		
		\section{Proofs of Theorems \ref{locfrtrans1} and \ref{locfrtrans2}}\label{app:proofs}
		
		Here we present the proofs of well-definedness, invertibility and composability of the localized frame transformations.
		
		\subsection{Proof of Theorem \ref{locfrtrans1}}\label{proof1}

		\begin{proof}
			Take $\Omega_1, \Omega_2 \in \S(\his \otimes \hio \otimes \hit)/\hspace{-3pt}\sim_G$ and write  $\Omega^{\R_i} = \Y^{\R_i}_*(\Omega)$ as usual. We need to show that for localizable $\E_1$ and any localizing sequence $(\omega_n)$ whenever $[\Omega_1^{\R_1}]_{\E_2} = [\Omega_2^{\R_1}]_{\E_2} $, i.e., whenever we have
			\[
			\tr[\Omega_1^{\R_1}(\F_\S \otimes \E_2(X))] =\tr[\Omega_2^{\R_1}(\F_\S \otimes \E_2(X))] \text{ for all } X \in \mathcal{B}(G), \F_\S \in \Eff(\his)
			\]
			we will also have $\Phi_{1\to2}^{\rm loc}(\Omega_1^{\R_1}) = \Phi_{1\to 2}^{\rm loc}(\Omega_2^{\R_1})$, i.e.,
				\[\lim_{n \to \infty}\tr[(\Omega_1^{\R_1} \otimes \omega_n)\Y^{\R_2}(\F_\S \otimes \E_1(X))] = \lim_{n \to \infty}\tr[(\Omega_2^{\R_1} \otimes \omega_n)\Y^{\R_2}(\F_\S \otimes \E_1(X))]\]
				
				for all $X \in \mathcal{B}(G), \F_\S \in \Eff(\his)$. We then calculate
			\begin{align*}
				\tr[\Phi_{1\to2}^{\rm loc}(\Omega_1^{\R_1})\F_\S \otimes \E_1(X)] &=
				\lim_{n \to \infty} \tr[(\Omega_1^{\R_1} \otimes \omega_n)\Y^{\R_2}(\F_\S \otimes \E_1(X))]\\ &=
				\lim_{n \to \infty} \tr[(\Omega_1^{\R_1} \otimes \omega_n)\int_G g.\F_\S \otimes \E_1(g.X)  \otimes d\E_2(g)]\\ &=
				\tr[\Omega_1^{\R_1} \int_G (\lim_{n \to \infty} \mu_{\omega_n}^{\E_1}(g.X))g.\F_\S \otimes d\E_2(g)]\\ &=
				\tr[\Omega_1^{\R_1} \int_G  \delta_e(g.X) g.\F_\S \otimes d\E_2(g)]\\
				&=\tr[\Omega_1^{\R_1} \int_G  \chi_{g.X}(e) g.\F_\S \otimes d\E_2(g)]\\ &= 
				\tr[\Omega_1^{\R_1} \int_G \chi_{X}(g^{-1}) g.\F_\S \otimes d\E_2(g)],
			\end{align*}
			where we have used that $\lim_{n \to \infty} \mu_{\omega_n}^{\E_1} = \delta_e$ and $\delta_e(g.X) = \chi_{g.X}(e) = \chi_X(g^{-1})$. Now we see that by hypothesis we can replace $\Omega_1^{\R_1}$ by $\Omega_2^{\R_1}$ and get the same number for any $X \in \mathcal{B}(G)$ and $\F_\S \in \Eff(\his)$. Running this calculation backwards gives the first claim, as the calculation does not depend on the choice of localizing sequence. To prove the second claim, we need to show that for arbitrary $\Omega \in \S(\his \otimes \hio \otimes \hit)/\hspace{-3pt}\sim_G$, $X \in \mathcal{B}(G)$ and $\F_\S \in \Eff(\his)$ we have
			\[
			\tr[\Phi_{1\to2}^{\rm loc}(\Omega^{\R_1}) \, \F_\S \otimes \E_1(X)] = \tr[\Omega^{\R_2} \, \F_\S \otimes \E_1(X)].
			\]
			We calculate
			\begin{align*}
				\tr[\Phi_{1\to2}^{\rm loc}(\Omega^{R_1}) \F_\S \otimes \E_1(X)] &=
				\tr[\Omega \int_G h.\left(\int_G \chi_X(g^{-1}) g.\F_\S \otimes d\E_2(g)\right) \otimes d\E_1(h)]\\ &=
				\tr[\Omega \int_G \int_G \chi_X(g^{-1})  hg.\F_\S \otimes d\E_2(hg) \otimes d\E_1(h)]. 
			\end{align*}
			Now performing the change of variables $l := hg$ in the $\E_2$ integral and exchanging the order of integration, we write
			\begin{align*}
				\tr[\Phi_{1\to2}^{\rm loc}(\Omega^{R_1}) \E_1(X) \otimes \F_\S] &=
				\tr[\Omega \int_G \int_G  \chi_X(l^{-1}h) l.\F_\S \otimes d\E_1(h) \otimes  d\E_2(l)] \\&=
				\tr[\Omega \int_G l.\F_\S \otimes \int_G \chi_X(l^{-1}h)d\E_1(h) \otimes d\E_2(l)].
			\end{align*}
			Since the $h$ variable appears only in the second tensor factor the $\E_1$ integral can be evaluated giving
			\[
			\int_G \chi_X(l^{-1}h) \, d\E_1(h) = \int_G \chi_{l.X}(h) \, d\E_1(h) = \E_1(l.X) = l.\E_1(X),
			\]

			and finally we arrive at
			\begin{align*}
				\tr[\Phi_{1\to2}^{\rm loc}(\Omega^{R_1}) \F_\S \otimes \E_1(X)] &=
				\tr[\Omega \int_G l.(\F_\S \otimes \E_1(X)) \otimes d\E_2(l)]\\
				&=\tr[\Omega \, \Y^{\R_2}(\F_\S \otimes \E_1(X))] =
				\tr[\Omega^{\R_2}(\F_\S \otimes \E_1(X))].
			\end{align*}
			
			To show the last claim, writing $(\eta_m)$ for a localizing sequence of $\R_2$, we calculate
			\begin{align*}
				&\tr[\Phi^{\rm loc}_{2 \to 1} \circ \Phi^{\rm loc}_{1 \to 2}(\Omega^{\R_1}) \, \F_\S \otimes \E_2(X)]\\
				&= \lim_{m \to \infty}\tr[\Y^{\R_1}_* \circ \mathcal{L}^{\R_2}_{\eta_m} \circ \Phi^{\rm loc}_{1 \to 2}(\Omega^{\R_1}) \, \F_\S \otimes \E_2(X)]\\
				&= \lim_{m \to \infty}\tr[\mathcal{L}^{\R_2}_{\eta_m} \circ \Phi^{\rm loc}_{1 \to 2}(\Omega^{\R_1}) \int_G g.(\F_\S \otimes \E_2(X)) \otimes d\E_1(g)]\\
				&= \lim_{m \to \infty}\tr[\Phi^{\rm loc}_{1 \to 2}(\Omega^{\R_1}) \int_G \mu^{\E_2}_{\eta_m}(g.X) g.\F_\S \otimes d\E_1(g)]\\
				&= \lim_{m \to \infty}\lim_{n \to \infty}\tr[\Y^{\R_2}_* \circ \mathcal{L}^{\R_1}_{\omega_n}(\Omega^{\R_1}) \int_G \mu^{\E_2}_{\eta_m}(g.X) g.\F_\S \otimes d\E_1(g)]\\
				&= \lim_{m \to \infty}\lim_{n \to \infty}\tr[\mathcal{L}^{\R_1}_{\omega_n}(\Omega^{\R_1}) \int_G h.\left(\int_G \mu^{\E_2}_{\eta_m}(g.X) g.\F_\S \otimes d\E_1(g) \right) \otimes d\E_2(h)]\\
				&= \lim_{m \to \infty}\lim_{n \to \infty}\tr[\mathcal{L}^{\R_1}_{\omega_n}(\Omega^{\R_1}) \int_G \int_G \mu^{\E_2}_{\eta_m}(g.X) hg.\F_\S \otimes d\E_1(hg) \otimes d\E_2(h)]\\
				&= \lim_{m \to \infty}\lim_{n \to \infty}\tr[\Omega^{\R_1} \int_G \int_G \mu^{\E_2}_{\eta_m}(g.X) hg.\F_\S \, d\mu^{\E_1}_{\omega_n}(hg) \otimes d\E_2(h)]\\
				&= \lim_{m \to \infty}\tr[\Omega^{\R_1} \int_G \mu^{\E_2}_{\eta_m}(h^{-1}.X)  \F_\S \otimes d\E_2(h)]\\
				&= \tr[\Omega^{\R_1} \, \F_\S \otimes \int_G \chi_X(h) \, d\E_2(h)] = \tr[\Omega^{\R_1} \F_\S \otimes \E_2(X)],
			\end{align*}
			where we have used $\lim_{n \to \infty} \mu^{\E_1}_{\omega_n}(gh) = \delta_e(gh) = \delta_{g^{-1}}(h)$ and $\lim_{m \to \infty} \mu^{\E_2}_{\eta_m}(h^{-1}.X) = \chi_X(h)$. From commutativity it follows that the map $\Phi_{1 \to 2}^{\rm loc}: \S(\his \otimes \hit)^{\R_1}_{\E_2} \to \S(\his \otimes \hio)^{\R_2}_{\E_1}$ is well-defined in the sense that taking the limit $n \to \infty$ does not take the outcome out of the codomain. Since $\Y^{\R_2}_*$ and $\mathcal{L}^{\R_1}_{\omega_n}$ are linear, we have a state space map.
		\end{proof}
		
		\subsection{Proof of Theorem \ref{locfrtrans2}}\label{proof2}

		\begin{proof}
			Writing $(\omega_n)$ for a localizing sequence of $\R_1$ and $(\eta_m)$ for that of $\R_2$ as before, we calculate
			\begin{align*}
				&\tr[\Phi^{\rm loc}_{2 \to 3} \circ \Phi^{\rm loc}_{1 \to 2}(\Omega^{\R_1}) \,
				\F_\S \otimes \E_1(X) \otimes \E_2(Y)]\\
				&= \lim_{m \to \infty} \tr[\Y^{\R_3}_* \circ \mathcal{L}_{\eta_m}^{\R_2} \circ \Phi^{\rm loc}_{1 \to 2}(\Omega^{\R_1})\F_\S \otimes \E_1(X) \otimes \E_2(Y)]\\
				&= \lim_{m \to \infty} \tr[\mathcal{L}_{\eta_m}^{\R_2} \circ \Phi^{\rm loc}_{1 \to 2}(\Omega^{\R_1})\int_G g.(\F_\S \otimes \E_1(X) \otimes \E_2(Y)) \otimes d\E_3(g)]\\
				&= \lim_{m \to \infty} \tr[\Phi^{\rm loc}_{1 \to 2}(\Omega^{\R_1})\int_G \mu_{\eta_m}^{\E_2}(g.Y) g.(\F_\S \otimes \E_1(X)) \otimes d\E_3(g)]\\
				&= \lim_{m \to \infty} \lim_{n \to \infty} \tr[\Y^{\R_2}_* \circ \mathcal{L}_{\omega_n}^{\R_1}(\Omega^{\R_1})\int_G \mu_{\eta_m}^{\E_2}(g.Y) g.(\F_\S \otimes \E_1(X)) \otimes d\E_3(g)]\\
				&= \lim_{m \to \infty} \lim_{n \to \infty} \tr[\mathcal{L}_{\omega_n}^{\R_1}(\Omega^{\R_1})\int_G h.\left(\int_G \mu_{\eta_m}^{\E_2}(g.Y) g.(\F_\S \otimes \E_1(X)) \otimes d\E_3(g)\right) \otimes d\E_2(h)]\\
				&= \lim_{m \to \infty} \lim_{n \to \infty} \tr[\mathcal{L}_{\omega_n}^{\R_1}(\Omega^{\R_1})\int_G \int_G \mu_{\eta_m}^{\E_2}(g.Y) hg.(\F_\S \otimes \E_1(X)) \otimes d\E_3(hg) \otimes d\E_2(h)].
			\end{align*}
			If we now change the integration variable in the $\E_3$ integral for $g' :=hg$ and change the order of integration we can write the operator above as
			\begin{align*}
				&\int_G \int_G \mu_{\eta_m}^{\E_2}(g.Y) hg.(\F_\S \otimes \E_1(X)) \otimes d\E_3(hg) \otimes d\E_2(h)\\
				&= \int_G g'.(\F_\S \otimes \E_1(X)) \otimes \int_G \mu_{\eta_m}^{\E_2}(h^{-1}g'.Y) d\E_2(h) \otimes d\E_3(g').
			\end{align*}
			Exchanging the order of limits and taking $m \to \infty$ the $\E_2$ integral can be evaluated giving
			\[
			\lim_{m \to \infty}\int_G \mu_{\eta_m}^{\E_2}(h^{-1}g'.Y) \, d\E_2(h) = \int_G \chi_{g'.Y}(h) \, d\E_2(h) = \E_2(g'.Y) = g'.\E_2(Y).
			\]
			We then get
			\begin{align*}
				&\tr[\Phi^{\rm loc}_{2 \to 3} \circ \Phi^{\rm loc}_{1 \to 2}(\Omega^{\R_1})
				\F_\S \otimes \E_1(X) \otimes \E_2(Y)]\\
				&= \lim_{m \to \infty} \tr[\mathcal{L}_{\omega_n}^{\R_1}(\Omega^{\R_1})\int_G g'.(\F_\S \otimes \E_1(X) \otimes \E_2(Y)) \otimes d\E_3(g')]\\
				&= \lim_{m \to \infty} \tr[\Y^{\R_3}_* \circ \mathcal{L}_{\omega_n}^{\R_1}(\Omega^{\R_1})\F_\S \otimes \E_1(X) \otimes \E_2(Y)]\\
				&= \tr[\Phi_{1\to 3}^{\rm loc}(\Omega^{\R_1})\F_\S \otimes \E_1(X) \otimes \E_2(Y)].
			\end{align*}
			Since $X,Y \in \mathcal{B}(G)$ and $\F_\S \in \Eff(\his)$ were arbitrary this completes the proof.
		\end{proof}
		
	\end{document}